\documentclass[
aps,
twocolumn,    %
showpacs,
preprintnumbers,
amsmath,
amssymb,
amsfonts,
superscriptaddress,
openany,
nofootinbib
]{revtex4-1}
\pdfoutput=1

\usepackage{amsmath,amssymb,amsthm,
graphicx,geometry,
lipsum,
mathtools,
tikz,multirow}
\usepackage{bm,color}
\usetikzlibrary{shapes.arrows,arrows.meta,fadings,shadows.blur}

\usepackage[utf8]{inputenc}
\usepackage{amsfonts}
\usepackage{amsbsy}
\usepackage{url}
\usepackage{enumerate}
\usepackage[justification=raggedright]{caption}
\usepackage{hhline}
\usepackage{graphicx,color}
\usepackage{amssymb,amsmath}
\usepackage{slashed}
\usepackage{hyperref}
\usepackage{braket}
\usepackage{simplewick}
\usepackage{ascmac}
\usepackage{latexsym}
\usepackage{pifont}          
\usepackage{bm}
\usepackage{comment}
\usepackage[normalem]{ulem}
\usepackage{cancel}
\usepackage{etoolbox} 			%%% Provides ability to hook into thebibliography

\makeatletter         
\newcounter{firstbib} 
\apptocmd{\thebibliography}{ \setcounter{NAT@ctr}{\value{firstbib}} }{}{} 
\AtEndEnvironment{thebibliography}{ \setcounter{firstbib}{\value{NAT@ctr}}}
\makeatother
\newtheorem{theorem}{Theorem}[section]
\newtheorem{lemma}[theorem]{Lemma}
\newtheorem{proposition}[theorem]{Proposition}

\newcommand{\SO}{\mathit{SO}}

\usepackage{ulem}

\begin{document}

\title{
%How to Tie %Topologically and Dynamically 
%Absolutely Stable 
%Non-Abelian 
%Quantum
%Knots %and Links 
%in Ultracold Gases
%\\ 
% Non-Abelian 
Quantum
Knots that Never Come Untied
}
\author{
Michikazu Kobayashi}
\affiliation{
School of Engineering Science, Kochi University of Technology, Miyanoguchi 185, Tosayamada, Kami, Kochi 782-5232, Japan
}
\affiliation{
International Institute for Sustainability with Knotted Chiral Meta Matter (WPI-SKCM$^2$),
Hiroshima University, 1-3-1 Kagamiyama, Higashi-Hiroshima, Hiroshima 739-8526, Japan
}
%\affiliation{
%International Institute for Sustainability with Knotted Chiral Meta Matter(SKCM$^2$), Hiroshima University, 1-3-2 Kagamiyama, Higashi-Hiroshima, Hiroshima 739-8511, Japan
%}
\author{
Yuta Nozaki}
\affiliation{
Faculty of Environment and Information Sciences, Yokohama National University, 
79-7 Tokiwadai, Hodogaya-ku, Yokohama, 240-8501, 
Japan
}
\affiliation{
International Institute for Sustainability with Knotted Chiral Meta Matter (WPI-SKCM$^2$), Hiroshima University, 
1-3-1 Kagamiyama, Higashi-Hiroshima, Hiroshima 739-8526, Japan
}
\author{
Yuya Koda}
\affiliation{
Department of Mathematics (Hiyoshi Campus), Keio University, 4-1-1 Hiyoshi, Yokohama, Kanagawa 223-8521, Japan}
\affiliation{
Research and Education Center for Natural Sciences, Keio University, 4-1-1 Hiyoshi, Yokohama, Kanagawa 223-8521, Japan}
\affiliation{
International Institute for Sustainability with Knotted Chiral Meta Matter (WPI-SKCM$^2$), Hiroshima University, 
1-3-1 Kagamiyama, Higashi-Hiroshima, Hiroshima 739-8526, Japan
}
\author{
Muneto Nitta}
\affiliation{Department of Physics,  Keio University, 4-1-1 Hiyoshi, Yokohama, Kanagawa 223-8521, Japan}
\affiliation{
Research and Education Center for Natural Sciences, Keio University, 4-1-1 Hiyoshi, Yokohama, Kanagawa 223-8521, Japan}
\affiliation{
International Institute for Sustainability with Knotted Chiral Meta Matter (WPI-SKCM$^2$), Hiroshima University, 
1-3-1 Kagamiyama, Higashi-Hiroshima, Hiroshima 739-8526, Japan
}

%\date{}

\begin{abstract}

Lord Kelvin proposed that atoms form hydrodynamic vortex knots.
However, they typically untie through reconnections, {\it i.~e.}, local cut-and-slice events, unlike stable vortex unknots such as smoke rings. The same holds in superfluids—quantum fluids with zero viscosity—where vortices have quantized circulation, making them topologically stable. For over 150 years, hydrodynamically stable vortex knots have been sought both experimentally and theoretically. 
Here, we present the first demonstration of hydrodynamically stable vortex knots and links in experimentally realizable Bose-Einstein condensates of ultracold atomic gases and confirm it through dynamic simulations. Our method creates stable knotted vortex structures in systems where reconnections are prohibited, with potential relevance to neutron star interiors. Additionally, we anticipate our mathematical framework could have applications in quantum computation, quantum turbulence, and DNA dynamics, particularly where reconnections are restricted.
% \textcolor{red}{(133 words)}

\end{abstract}

\maketitle

%%%%%%%%%%%%%%%%%%%%%%%%%%%%
%\section{Introduction}

Following Helmholtz's seminal work on vortex motion \cite{Helmholtz+1858+25+55}, Lord Kelvin proposed in his pioneering work \cite{Thomson:1869} that atoms, the fundamental building blocks of matter, are vortex knots in the aether (see the top panel of Fig.~\ref{fig:main-results}). Although this conjecture garnered significant attention from scientists, it was ultimately ruled out with the advent of modern physics, which established that atoms consist of quarks and gluons. However, Kelvin's idea spurred the development of knot theory \cite{tait1898knots} as a distinct mathematical discipline. Since then, knots have been observed in numerous physical, chemical, and biological systems \cite{kauffman1991knots}, including hydrodynamics \cite{Moffatt_1969,10.1063/1.881574,Ricca:2009,Kleckner:2013,Arnold:2021}, DNA and proteins \cite{Wasserman}, superconductors \cite{Babaev:2001zy,Rybakov:2018ktd}, Bose-Einstein condensates (BECs) \cite{PhysRevE.85.036306,Kleckner_2016,Kawaguchi:2008xi,Hall:2016}, %$^3$He 
superfluids \cite{Volovik:1977}, liquid crystals \cite{PhysRevLett.110.237801,Ackerman:2015,Ackerman:2017,Tai:2019,RevModPhys.84.497,Smalyukh:2020zin}, colloids \cite{TKALEC,Martinez}, magnets \cite{Kent:2020jvm}, active matter \cite{Shankar2022}, nonlinear science \cite{Maucher:2018wco}, optics \cite{Dennis2010}, electromagnetism \cite{Kedia:2013bw,Arrayas:2017sfq}, elementary particle physics \cite{Eto:2024hwn}, and quantum chromodynamics \cite{Faddeev:1996zj}.

Kelvin's idea of vortex knots as atoms was rooted in his circulation theorem \cite{Thomson:1869}, which posits that knotted or linked vortex lines remain preserved over time in inviscid and barotropic perfect fluids. This led to extensive studies of vortex knots in hydrodynamics \cite{Moffatt_1969,10.1063/1.881574,Ricca:2009,Kleckner:2013} (see \cite{Arnold:2021} for a review), with experimental realizations in water \cite{Kleckner:2013}. However, these vortex knots typically untie due to viscosity, prompting the question of whether stable hydrodynamic vortex knots exist.

In quantum physics, vortices exhibit quantized circulation in superfluids, such as helium superfluids and BECs of ultracold atomic gases. These quantum hydrodynamic fluids have zero viscosity, which allows for the stability of minimally quantized vortices, similar to how hydrogen atoms are stable in quantum mechanics. Nevertheless, quantum vortex knots remain unstable, typically breaking apart through reconnections \cite{PhysRevE.85.036306}. The number of reconnections required to untie a vortex knot is mathematically related to a knot invariant \cite{Kleckner_2016}.

\begin{figure*}[htb]
\centering
\includegraphics[width=0.85\textwidth]{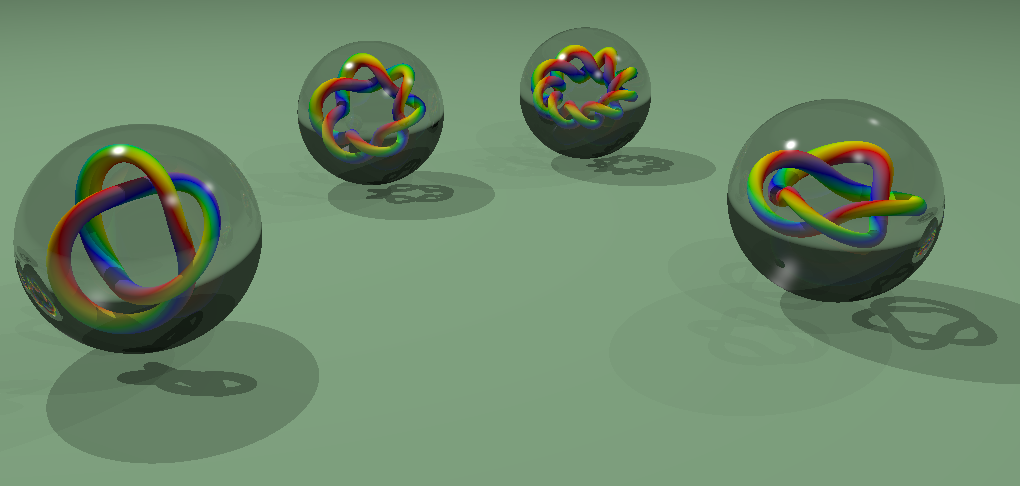}\\[10pt]
\begin{minipage}{0.03\linewidth}
    \begin{tikzpicture}
        \draw[->,>=Stealth,line width=1.5](0,0)--(0,7.2) node[midway,above,rotate=90]{time};
    \end{tikzpicture}
\end{minipage}
\begin{minipage}{0.24\linewidth}
  \centering
  (a) \\
  \includegraphics[height=2.0\linewidth]{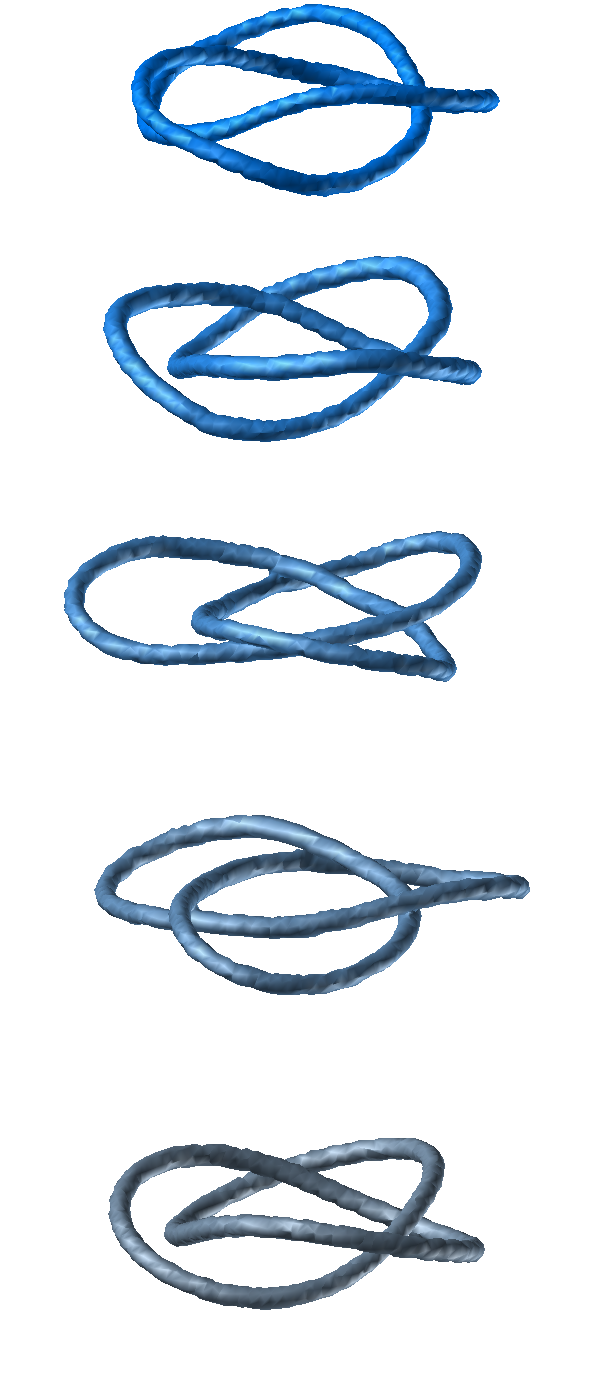}
\end{minipage}\hspace{-15pt}
\begin{minipage}{0.24\linewidth}
  \centering
  (b) \\
  \includegraphics[height=2.0\linewidth]{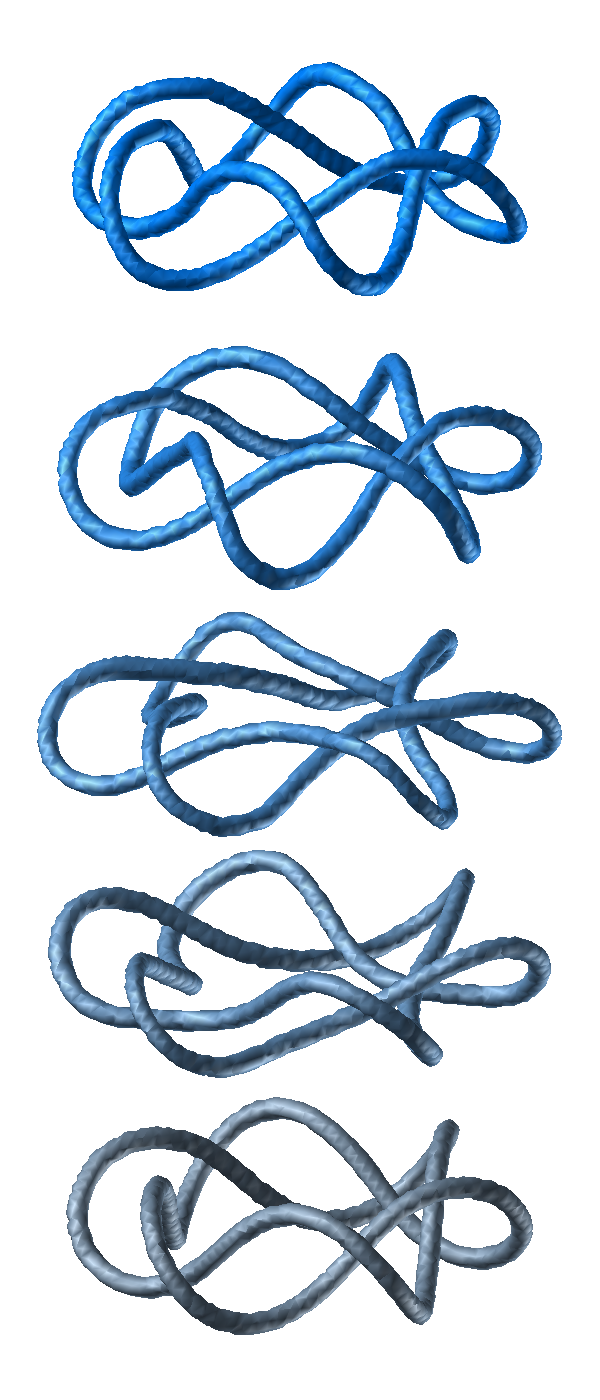}
\end{minipage}\hspace{-15pt}
\begin{minipage}{0.24\linewidth}
  \centering
  (c) \\
  \includegraphics[height=2.0\linewidth]{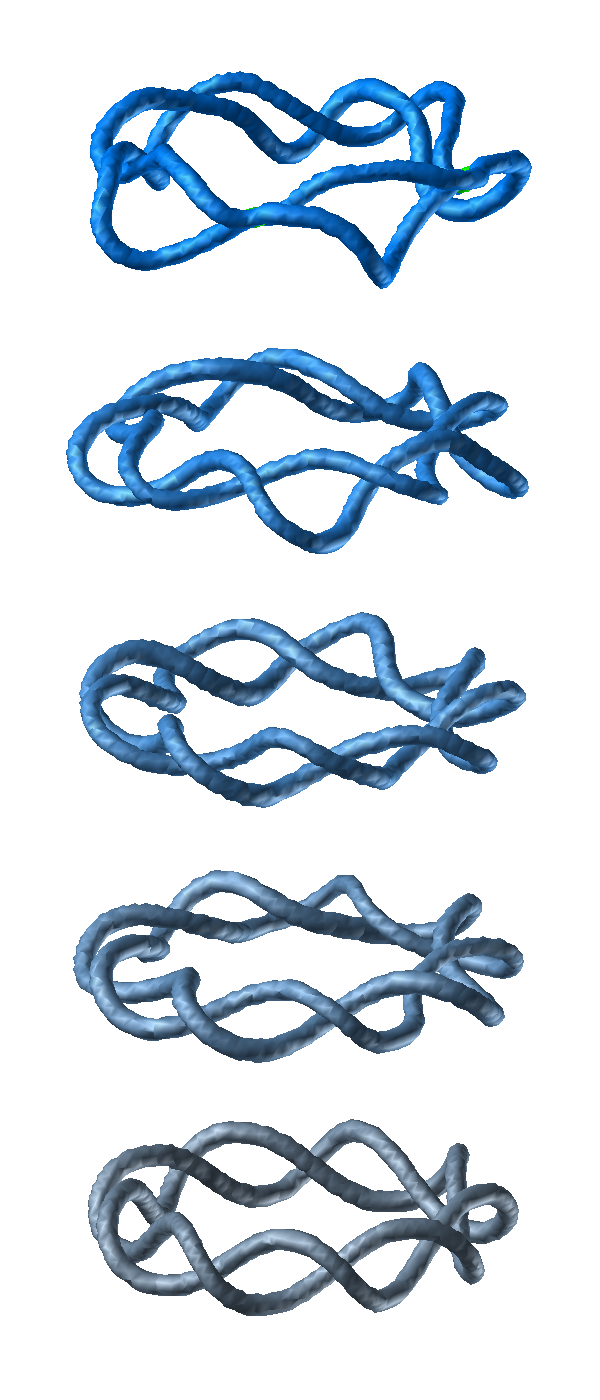}
\end{minipage}\hspace{-15pt}
\begin{minipage}{0.24\linewidth}
  \centering
  (d) \\
  \includegraphics[height=2.0\linewidth]{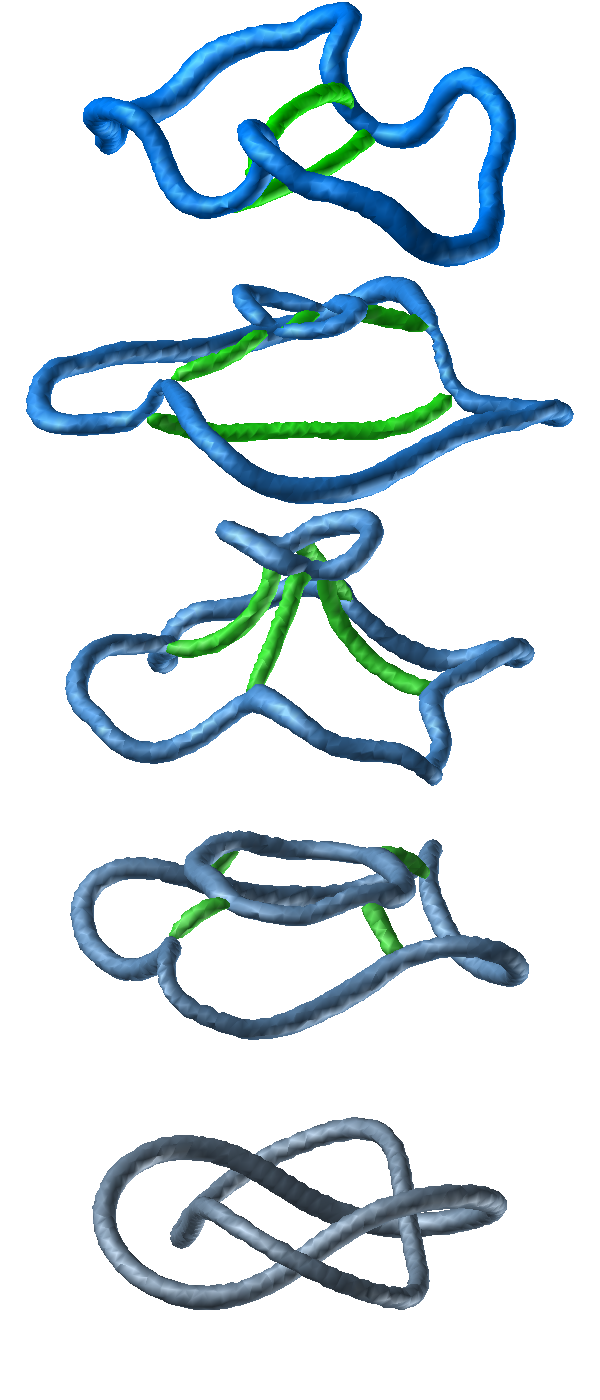}
\end{minipage}
\caption{
\label{fig:main-results}
Upper panel: Image of Lord Kelvin's vortex atom model, where atoms are represented as vortex knots in the aether fluid.
Lower panel: Dynamics of stable and metastable vortex knots and links in spin-2 spinor BECs, corresponding to the vortex atom model shown in the upper panel.
(a): A trefoil knot ($(2,3)$-torus knot),
(b): A $(2,6)$-torus link,
(c): A $(2,9)$-torus knot in the cyclic phase, and
(d): A $(2,4)$-torus link in the nematic phase
(see Supplementary Movies 4, 9, 11, and 19, respectively).
All vortex knots and links move from bottom to top over time at constant velocities. The blue curves represent hydrodynamic vortices, while the green curves depict non-hydrodynamic (rung) vortices.
(a) is the unique stable knot.
(b) and (c) are constructed by repeating the braid for the trefoil in panel (a) twice and three times, respectively, and taking their closures.
(d) is metastable.}
\end{figure*}
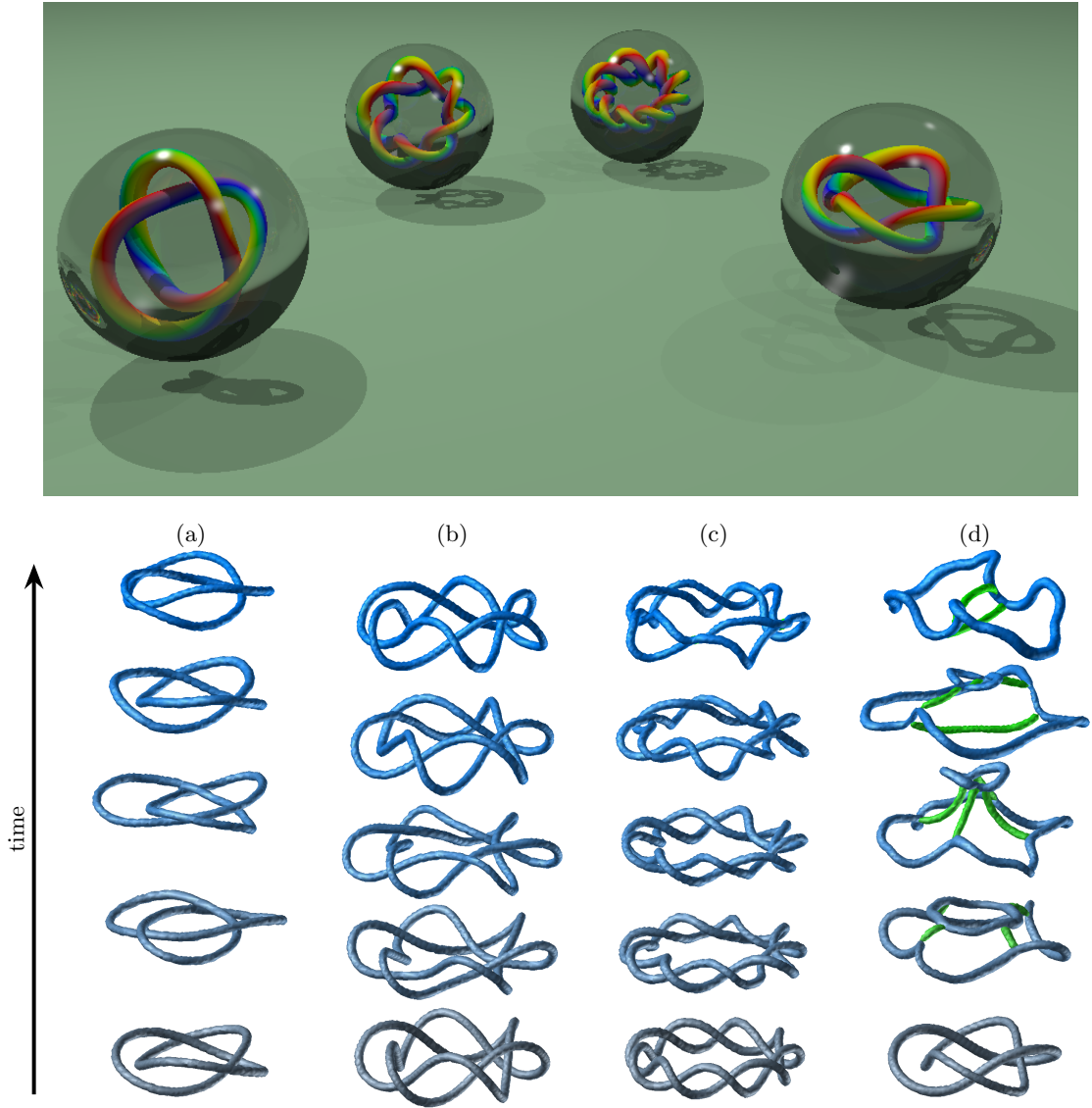
Quantized vortices are characterized by the fundamental group of the order parameter manifold of the system. When this group is non-Abelian, its elements are generally non-commutative, meaning that vortices do not reconnect in the traditional sense and instead create a vortex bridge called a rung when they collide \cite{Poenaru1977,Mermin:1979zz}. Examples of non-Abelian vortices occur in biaxial nematic liquid crystals and in spinor BECs which are BECs with spin degrees of freedom \cite{Kawaguchi:2012ii} and experimentally realized in ultracold $^{87}$Rb atoms \cite{PhysRevA.80.042704}. Another example is $P$-wave neutron superfluids in neutron stars \cite{Chamel2017,Haskell:2017lkl,Sedrakian:2018ydt}. In spinor BECs \cite{Semenoff:2006vv,Kobayashi:2008pk,Borgh:2016cco,Mawson:2018klj,Kawaguchi:2012ii}, non-Abelian vortices form rungs during collisions, and while vortex knots have been theoretically discussed \cite{Annala_2022}, no conclusive statements on their stability have been made.

Here, we demonstrate for the first time the existence of stable vortex knots and links by exhaustively analyzing all possible knots and links in experimentally realizable BECs, specifically spinor BECs with total spin up to two. We confirm their stability through extensive numerical simulations. As shown in Fig.~\ref{fig:main-results}, a unique stable vortex knot (see Supplementary Movie 4) exists in the cyclic phase of a spin-2 BEC. Additionally, we observe several metastable knots and links (see Supplementary Movies 9, 11, and 19) in both the cyclic and nematic phases of spin-2 BECs. These results mark a significant breakthrough in the study of vortex dynamics, providing the first concrete evidence of hydrodynamically stable vortex knots in quantum fluids.
% \red{(486 + 4 - 2 + 1 words)}

% \sout{Our findings open new avenues for further research into the behavior of vortex knots in other quantum systems. Potential applications include quantum turbulence, quantum computation, and even biological systems like DNA, where knot reconnections may play a crucial role. Moreover, these stable vortex structures may provide insights into astrophysical phenomena, such as the dynamics of neutron stars, where similar vortex behavior could occur in the presence of a superfluid.
% \red{(554+4-2 words)}}

%%%%%%%%%%%%%%%%%%
\section*{Stability of vortex knots and links}

\begin{figure}[htb]
\centering
\begin{minipage}{0.49\linewidth}
  \centering
  (a) \\[-5pt]
  \begin{tikzpicture}
  \draw[->, >=Stealth](0,0)--(0,5) node[right]{$t$};
  \draw[->, >=Stealth](0,0)--(1,0) node[right]{$x$};
  \draw[->, >=Stealth](0,0)--(0.6,0.4) node[right]{$y$};
  \filldraw[red] (1.5,0.3) circle(3pt) node[above,left]{$+$};
  \filldraw[blue] (3.0,0.3) circle(3pt) node[above,right]{$+$};
  \draw[->, >=Triangle, line width=1.5](2.85,0.35) arc(40:150:0.75 and 0.55);
  \draw[->, >=Triangle, line width=1.5](1.65,0.25) arc(220:330:0.75 and 0.55);
  \draw[blue, line width=1.5, domain = 0:80] plot({2.25+0.75*cos(\x)}, {0.3+4.5*\x/360});
  \draw[blue, line width=1.5, domain = 100:360] plot({2.25+0.75*cos(\x)}, {0.3+4.5*\x/360});
  \draw[red, line width=1.5, domain = 0:260] plot({2.25-0.75*cos(\x)}, {0.3+4.5*\x/360});
  \draw[red, line width=1.5, domain = 280:360] plot({2.25-0.75*cos(\x)}, {0.3+4.5*\x/360});
\end{tikzpicture}
\end{minipage}
\begin{minipage}{0.49\linewidth}
  \centering
  (b) \\[-5pt]
  \begin{tikzpicture}
  \draw[->, >=Stealth](0,0)--(0,5) node[right]{$-t$};
  \draw[->, >=Stealth](0,0)--(1,0) node[right]{$x$};
  \draw[->, >=Stealth](0,0)--(0.6,0.4) node[right]{$y$};
  \filldraw[red] (1.5,0.3) circle(3pt) node[above,left]{$+$};
  \filldraw[blue] (3.0,0.3) circle(3pt) node[above,right]{$+$};
  \draw[->, >=Triangle, line width=1.5](1.65,0.35) arc(140:30:0.75 and 0.55);
  \draw[->, >=Triangle, line width=1.5](2.85,0.25) arc(320:210:0.75 and 0.55);
  \draw[blue, line width=1.5, domain = 0:260] plot({2.25+0.75*cos(\x)}, {0.3+4.5*\x/360});
  \draw[blue, line width=1.5, domain = 280:360] plot({2.25+0.75*cos(\x)}, {0.3+4.5*\x/360});
  \draw[red, line width=1.5, domain = 0:80] plot({2.25-0.75*cos(\x)}, {0.3+4.5*\x/360});
  \draw[red, line width=1.5, domain = 100:360] plot({2.25-0.75*cos(\x)}, {0.3+4.5*\x/360});
\end{tikzpicture}
\end{minipage}
\caption{
\label{fig:braid_xyt}
Dynamics of two point vortices in the 2-dimensional $xy$-plane and a vortex braid in the $(2+1)$-dimensional $xyt$-spacetime.
When two point vortices have hydrodynamic circulations with positive orientation, they rotate around each other counterclockwise (clockwise) toward the future (past) in panel (a) (panel (b)), forming a positive (negative) braid in $xyt$-spacetime.
The braid's orientation reverses for vortices with negative hydrodynamic circulation.
}
\end{figure}
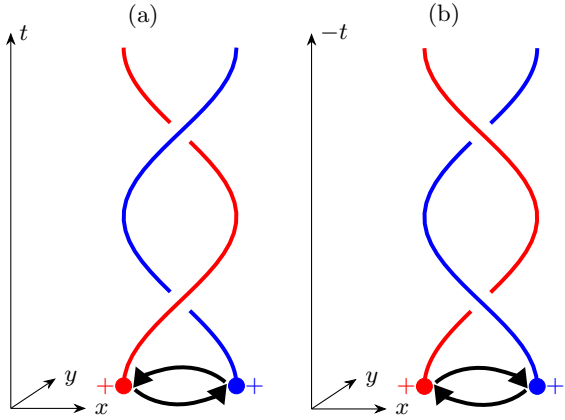
In order to construct vortex knots and links, we start by considering vortex braids that correspond to those knots and links as their closures. To begin, let us consider a hypothetical two-dimensional plane where vortices are point-like. Two point vortices with hydrodynamic circulations of the same orientation rotate around each other, forming a braid in $(2+1)$-dimensional spacetime (see Fig.~\ref{fig:braid_xyt}). Two point vortices with positive circulations rotate counterclockwise in the future direction, as shown in Fig.~\ref{fig:braid_xyt}(a), or clockwise in the past, as seen in Fig.~\ref{fig:braid_xyt}(b), forming a positive or negative braid, respectively.

\begin{figure*}[htb]
\centering
\begin{minipage}{0.49\linewidth}
  \centering
  (a) \\[-5pt]
  \begin{tikzpicture}
  \draw[->, >=Stealth](0,0)--(0,5) node[right]{$z$};
  \draw[->, >=Stealth](0,0)--(1,0) node[right]{$x$};
  \draw[->, >=Stealth](0,0)--(0.6,0.4) node[right]{$y$};
  % \filldraw[red] (1.5,0.3) circle(3pt) node[above,left]{$+$};
  % \filldraw[blue] (3.0,0.3) circle(3pt) node[above,right]{$+$};
  % \draw[->, >=Triangle, line width=1.5](2.85,0.35) arc(40:150:0.75 and 0.55);
  % \draw[->, >=Triangle, line width=1.5](1.65,0.25) arc(220:330:0.75 and 0.55);
  \draw[blue, line width=1.5, domain = 0:80] plot({2.25+0.75*cos(\x)}, {0.3+4.5*\x/360});
  \draw[blue, line width=1.5, domain = 100:360] plot({2.25+0.75*cos(\x)}, {0.3+4.5*\x/360});
  \draw[red, line width=1.5, domain = 0:260] plot({2.25-0.75*cos(\x)}, {0.3+4.5*\x/360});
  \draw[red, line width=1.5, domain = 280:360] plot({2.25-0.75*cos(\x)}, {0.3+4.5*\x/360});
  \draw[->, >=Stealth, line width = 2] (3.5,2.5)--(4.5,2.5);
  \path(2.25-0.75,4.8)--(2.25+0.75,4.8) node[midway,above]{$t=0$};

  \begin{scope}[xshift=100]
    \draw[blue, line width=1.5, domain = 0:25] plot({2.25+0.75*cos(\x+55)}, {0.3+4.5*\x/360});
    \draw[blue, line width=1.5, domain = 45:360] plot({2.25+0.75*cos(\x+55)}, {0.3+4.5*\x/360});
    \draw[red, line width=1.5, domain = 0:205] plot({2.25-0.75*cos(\x+55)}, {0.3+4.5*\x/360});
    \draw[red, line width=1.5, domain = 225:360] plot({2.25-0.75*cos(\x+55)}, {0.3+4.5*\x/360});
    \path(2.25-0.75,4.8)--(2.25+0.75,4.8) node[midway,above]{$t>0$};

    \draw[->, >=Triangle, dashed, line width=1] (3.5,0.3+3.25)--(3.5,0.3+1.25);
  \end{scope}
\end{tikzpicture}
\end{minipage}
\begin{minipage}{0.49\linewidth}
  \centering
  (b) \\[-5pt]
  \begin{tikzpicture}
  \draw[->, >=Stealth](0,0)--(0,5) node[right]{$z$};
  \draw[->, >=Stealth](0,0)--(1,0) node[right]{$x$};
  \draw[->, >=Stealth](0,0)--(0.6,0.4) node[right]{$y$};
  % \filldraw[red] (1.5,0.3) circle(3pt) node[above,left]{$+$};
  % \filldraw[blue] (3.0,0.3) circle(3pt) node[above,right]{$+$};
  % \draw[->, >=Triangle, line width=1.5](2.85,0.35) arc(40:150:0.75 and 0.55);
  % \draw[->, >=Triangle, line width=1.5](1.65,0.25) arc(220:330:0.75 and 0.55);
  \draw[blue, line width=1.5, domain = 0:260] plot({2.25+0.75*cos(\x)}, {0.3+4.5*\x/360});
  \draw[blue, line width=1.5, domain = 280:360] plot({2.25+0.75*cos(\x)}, {0.3+4.5*\x/360});
  \draw[red, line width=1.5, domain = 0:80] plot({2.25-0.75*cos(\x)}, {0.3+4.5*\x/360});
  \draw[red, line width=1.5, domain = 100:360] plot({2.25-0.75*cos(\x)}, {0.3+4.5*\x/360});
  \draw[->, >=Stealth, line width = 2] (3.5,2.5)--(4.5,2.5);
  \path(2.25-0.75,4.8)--(2.25+0.75,4.8) node[midway,above]{$t=0$};

  \begin{scope}[xshift=100]
    \draw[blue, line width=1.5, domain = 0:315] plot({2.25+0.75*cos(\x-55)}, {0.3+4.5*\x/360});
    \draw[blue, line width=1.5, domain = 330:360] plot({2.25+0.75*cos(\x-55)}, {0.3+4.5*\x/360});
    \draw[red, line width=1.5, domain = 0:135] plot({2.25-0.75*cos(\x-55)}, {0.3+4.5*\x/360});
    \draw[red, line width=1.5, domain = 155:360] plot({2.25-0.75*cos(\x-55)}, {0.3+4.5*\x/360});
    \path(2.25-0.75,4.8)--(2.25+0.75,4.8) node[midway,above]{$t>0$};

    \draw[->, >=Triangle, dashed, line width=1] (3.5,0.3+1.25)--(3.5,0.3+3.25);
  \end{scope}
\end{tikzpicture}
\end{minipage}
\caption{
\label{fig:braid_xyz}
Dynamics of positive (negative) braiding of two vortex lines in panel (a) (panel (b)).
Dashed arrows show the direction of the braid evolution.
}
\end{figure*}
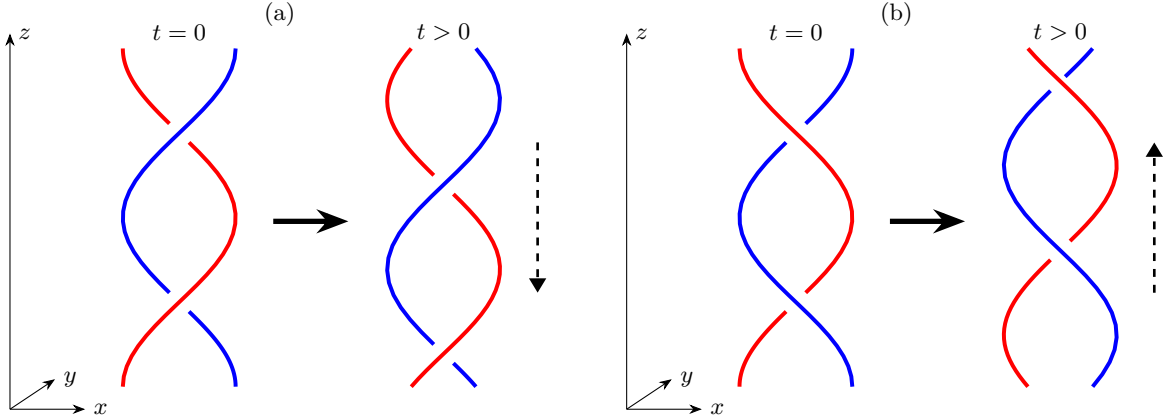
We can naturally extend this concept of vortex braiding into 3-dimensional space, similar to the braids seen in $(2+1)$-dimensional spacetime. Fig.~\ref{fig:braid_xyz}(a) shows the positive braiding dynamics of two vortex lines, where the direction of braiding mirrors Fig.~\ref{fig:braid_xyt}(a), but with the $z$-axis replaced by time. In each constant $z$-slice, two-point vortices at their intersections rotate counterclockwise over time, causing the braid to move in the negative direction of the $z$-axis. Conversely, the negative braiding of vortex lines, as shown in Fig.~\ref{fig:braid_xyz}(b), moves in the positive direction of the $z$-axis. As a result, vortex braids in 3-dimensional space propagate in a fixed direction.

\begin{figure*}[htb]
  \centering
  \begin{minipage}[t]{0.245\linewidth}
    \centering (a) \\[5pt]
    \begin{tikzpicture}
    \begin{scope}[rotate=60]
        \draw[blue, line width = 1.5, domain = 0:110] plot({(1+0.25*sin(1.5*\x))*cos(\x)}, {(1+0.25*sin(1.5*\x))*sin(\x)});
        \draw[green, line width = 1.5, domain = 130:295] plot({(1+0.25*sin(1.5*\x))*cos(\x)}, {(1+0.25*sin(1.5*\x))*sin(\x)});
        \draw[green, line width = 1.5, domain = 305:350] plot({(1+0.25*sin(1.5*\x))*cos(\x)}, {(1+0.25*sin(1.5*\x))*sin(\x)});
        \draw[red, line width = 1.5, domain = 10:230] plot({(1-0.25*sin(1.5*\x))*cos(\x)}, {(1-0.25*sin(1.5*\x))*sin(\x)});
        \draw[blue, line width = 1.5, domain = 250:295] plot({(1-0.25*sin(1.5*\x))*cos(\x)}, {(1-0.25*sin(1.5*\x))*sin(\x)});
        \draw[blue, line width = 1.5, domain = 305:360] plot({(1-0.25*sin(1.5*\x))*cos(\x)}, {(1-0.25*sin(1.5*\x))*sin(\x)});

        \draw[blue, ->, >=Triangle, line width = 1.5, domain = 40:50] plot({(1+0.25*sin(1.5*\x))*cos(\x)}, {(1+0.25*sin(1.5*\x))*sin(\x)});
        \draw[green, ->, >=Triangle, line width = 1.5, domain = 280:290] plot({(1+0.25*sin(1.5*\x))*cos(\x)}, {(1+0.25*sin(1.5*\x))*sin(\x)});
        \draw[red, ->, >=Triangle, line width = 1.5, domain = 160:170] plot({(1-0.25*sin(1.5*\x))*cos(\x)}, {(1-0.25*sin(1.5*\x))*sin(\x)});
    \end{scope}
    \draw[line width=1](0.2,0)--(0.4,0.3);
    \draw[line width=1](0.2,0)--(1.45,0)--(1.65,0.3);
    \draw[line width=1](0.4,0.3)--(0.6,0.3);
    \draw[line width=1](0.9,0.3)--(1.05,0.3);
    \draw[line width=1](1.35,0.3)--(1.65,0.3);
    \draw[->, >=Triangle, line width=1, dashed, domain=5:60] plot({1.55*cos(\x)}, {1.55*sin(\x)});
    \path ({1.55*cos(60)}, {1.55*sin(60)}) node[above=4,right=4]{$\theta$};
\end{tikzpicture}
  \end{minipage}
  \begin{minipage}[t]{0.245\linewidth}
    \centering (c) \\[5pt]
    \begin{tikzpicture}
  \begin{scope}[rotate=45]
    \draw[blue, line width = 1.5, domain = -80:-50] plot({(1+0.25*sin(2*\x))*cos(\x)}, {(1+0.25*sin(2*\x))*sin(\x)});
    \draw[blue, line width = 1.5, domain = -40:80] plot({(1+0.25*sin(2*\x))*cos(\x)}, {(1+0.25*sin(2*\x))*sin(\x)});
    \draw[green, line width = 1.5, domain = 100:260] plot({(1+0.25*sin(2*\x))*cos(\x)}, {(1+0.25*sin(2*\x))*sin(\x)});
    \draw[red, line width = 1.5, domain = 10:170] plot({(1-0.25*sin(2*\x))*cos(\x)}, {(1-0.25*sin(2*\x))*sin(\x)});
    \draw[violet, line width = 1.5, domain = 190:310] plot({(1-0.25*sin(2*\x))*cos(\x)}, {(1-0.25*sin(2*\x))*sin(\x)});
    \draw[violet, line width = 1.5, domain = 320:350] plot({(1-0.25*sin(2*\x))*cos(\x)}, {(1-0.25*sin(2*\x))*sin(\x)});

    \draw[blue, ->, >=Triangle, line width = 1.5, domain = 20:30] plot({(1+0.25*sin(2*\x))*cos(\x)}, {(1+0.25*sin(2*\x))*sin(\x)});
    \draw[green, ->, >=Triangle, line width = 1.5, domain = 200:210] plot({(1+0.25*sin(2*\x))*cos(\x)}, {(1+0.25*sin(2*\x))*sin(\x)});
    \draw[red, ->, >=Triangle, line width = 1.5, domain = 110:120] plot({(1-0.25*sin(2*\x))*cos(\x)}, {(1-0.25*sin(2*\x))*sin(\x)});
    \draw[violet, ->, >=Triangle, line width = 1.5, domain = 290:300] plot({(1-0.25*sin(2*\x))*cos(\x)}, {(1-0.25*sin(2*\x))*sin(\x)});
  \end{scope}
  \draw[line width=1](0.2,0)--(0.4,0.3);
  \draw[line width=1](0.2,0)--(1.45,0)--(1.65,0.3);
  \draw[line width=1](0.4,0.3)--(0.6,0.3);
  \draw[line width=1](0.9,0.3)--(1.05,0.3);
  \draw[line width=1](1.35,0.3)--(1.65,0.3);
  \draw[->, >=Triangle, line width=1, dashed, domain=5:60] plot({1.55*cos(\x)}, {1.55*sin(\x)});
  \path ({1.55*cos(60)}, {1.55*sin(60)}) node[above=4,right=4]{$\theta$};
\end{tikzpicture}
  \end{minipage}
  \begin{minipage}[t]{0.245\linewidth}
    \centering (e) \\[5pt]
    \begin{tikzpicture}
  \begin{scope}[rotate=45]
    \draw[blue, line width = 1.5, domain = -80:-50] plot({(1+0.25*sin(2*\x))*cos(\x)}, {(1+0.25*sin(2*\x))*sin(\x)});
    \draw[blue, line width = 1.5, domain = -40:80] plot({(1+0.25*sin(2*\x))*cos(\x)}, {(1+0.25*sin(2*\x))*sin(\x)});
    \draw[green, line width = 1.5, domain = 100:260] plot({(1+0.25*sin(2*\x))*cos(\x)}, {(1+0.25*sin(2*\x))*sin(\x)});
    \draw[red, line width = 1.5, domain = 170:10] plot({(1-0.25*sin(2*\x))*cos(\x)}, {(1-0.25*sin(2*\x))*sin(\x)});
    \draw[violet, line width = 1.5, domain = 190:310] plot({(1-0.25*sin(2*\x))*cos(\x)}, {(1-0.25*sin(2*\x))*sin(\x)});
    \draw[violet, line width = 1.5, domain = 320:350] plot({(1-0.25*sin(2*\x))*cos(\x)}, {(1-0.25*sin(2*\x))*sin(\x)});

    \draw[blue, ->, >=Triangle, line width = 1.5, domain = 20:30] plot({(1+0.25*sin(2*\x))*cos(\x)}, {(1+0.25*sin(2*\x))*sin(\x)});
    \draw[green, ->, >=Triangle, line width = 1.5, domain = 200:210] plot({(1+0.25*sin(2*\x))*cos(\x)}, {(1+0.25*sin(2*\x))*sin(\x)});
    \draw[red, ->, >=Triangle, line width = 1.5, domain = 110:100] plot({(1-0.25*sin(2*\x))*cos(\x)}, {(1-0.25*sin(2*\x))*sin(\x)});
    \draw[violet, ->, >=Triangle, line width = 1.5, domain = 290:280] plot({(1-0.25*sin(2*\x))*cos(\x)}, {(1-0.25*sin(2*\x))*sin(\x)});
  \end{scope}
  \draw[line width=1](0.2,0)--(0.4,0.3);
  \draw[line width=1](0.2,0)--(1.45,0)--(1.65,0.3);
  \draw[line width=1](0.4,0.3)--(0.6,0.3);
  \draw[line width=1](0.9,0.3)--(1.05,0.3);
  \draw[line width=1](1.35,0.3)--(1.65,0.3);
  \draw[->, >=Triangle, line width=1, dashed, domain=5:60] plot({1.55*cos(\x)}, {1.55*sin(\x)});
  \path ({1.55*cos(60)}, {1.55*sin(60)}) node[above=4,right=4]{$\theta$};
\end{tikzpicture}
  \end{minipage}
  \begin{minipage}[t]{0.245\linewidth}
    \centering (g) \\[5pt]
    \begin{tikzpicture}
  \begin{scope}[rotate=90]
  \draw[samples = 100, red,    line width = 1.5, domain = -205:-90]  plot({1.00*(1+0.55*sin(2*\x/3))*cos(\x)},     {1.00*(1+0.55*sin(2*\x/3))*sin(\x)});
  \draw[samples = 100, red,    line width = 1.5, domain = -80:35]  plot({1.00*(1+0.55*sin(2*\x/3))*cos(\x)},     {1.00*(1+0.55*sin(2*\x/3))*sin(\x)});
  \draw[samples = 100, violet, line width = 1.5, domain = 55:265]   plot({1.00*(1+0.55*sin(2*\x/3))*cos(\x)},     {1.00*(1+0.55*sin(2*\x/3))*sin(\x)});
  \draw[samples = 100, violet, line width = 1.5, domain = 275:305]   plot({1.00*(1+0.55*sin(2*\x/3))*cos(\x)},     {1.00*(1+0.55*sin(2*\x/3))*sin(\x)});
  \draw[samples = 100, green,  line width = 1.5, domain = 325:575]  plot({1.00*(1+0.55*sin(2*\x/3))*cos(\x)},     {1.00*(1+0.55*sin(2*\x/3))*sin(\x)});
  \draw[samples = 100, blue,   line width = 1.5, domain = -125:-95] plot({1.00*(1+0.55*sin(2*\x/3+120))*cos(\x)}, {1.00*(1+0.55*sin(2*\x/3+120))*sin(\x)});
  \draw[samples = 100, blue,   line width = 1.5, domain = -85:125] plot({1.00*(1+0.55*sin(2*\x/3+120))*cos(\x)}, {1.00*(1+0.55*sin(2*\x/3+120))*sin(\x)});

  \draw[samples = 100, red,    ->, >=Triangle, line width = 1.5, domain = 5:15]  plot({1.00*(1+0.55*sin(2*\x/3))*cos(\x)},     {1.00*(1+0.55*sin(2*\x/3))*sin(\x)});
  \draw[samples = 100, violet, ->, >=Triangle, line width = 1.5, domain = 125:135]   plot({1.00*(1+0.55*sin(2*\x/3))*cos(\x)},     {1.00*(1+0.55*sin(2*\x/3))*sin(\x)});
  \draw[samples = 100, green,  ->, >=Triangle, line width = 1.5, domain = 545:555]  plot({1.00*(1+0.55*sin(2*\x/3))*cos(\x)},     {1.00*(1+0.55*sin(2*\x/3))*sin(\x)});
  \draw[samples = 100, blue,   ->, >=Triangle, line width = 1.5, domain = -55:-45] plot({1.00*(1+0.55*sin(2*\x/3+120))*cos(\x)}, {1.00*(1+0.55*sin(2*\x/3+120))*sin(\x)});
  \end{scope}

  \draw[line width=1](0.2,0)--(0.4,0.3);
  \draw[line width=1](0.2,0)--(1.65,0)--(1.85,0.3);
  \draw[line width=1](0.4,0.3)--(0.45,0.3);
  \draw[line width=1](0.65,0.3)--(0.7,0.3);
  \draw[line width=1](0.90,0.3)--(1.35,0.3);
  \draw[line width=1](1.55,0.3)--(1.85,0.3);
  \draw[->, >=Triangle, line width=1, dashed, domain=5:60] plot({1.75*cos(\x)}, {1.75*sin(\x)});
  \path ({1.75*cos(60)}, {1.75*sin(60)}) node[above=4,right=4]{$\theta$};  
\end{tikzpicture}
  \end{minipage} \\[5pt]
  \begin{minipage}[t]{0.245\linewidth}
    \centering (b) \\[5pt]
      \begin{tikzpicture}
    \draw[green, line width = 1.5, domain = 0:50] plot({0.5*cos(1.5*\x)}, {3*\x/360});
    \draw[red, line width = 1.5, domain = 70:290] plot({0.5*cos(1.5*\x)}, {3*\x/360});
    \draw[blue, line width = 1.5, domain = 310:360] plot({0.5*cos(1.5*\x)}, {3*\x/360});
    \draw[blue, line width = 1.5, domain = 0:170] plot({-0.5*cos(1.5*\x)}, {3*\x/360});
    \draw[green, line width = 1.5, domain = 190:360] plot({-0.5*cos(1.5*\x)}, {3*\x/360});

    \draw[->, >=Triangle, blue, line width = 1.5, domain = 70:90] plot({-0.5*cos(1.5*\x)}, {3*\x/360});
    \draw[->, >=Triangle, red, line width = 1.5, domain = 190:210] plot({0.5*cos(1.5*\x)}, {3*\x/360});
    \draw[->, >=Triangle, green, line width = 1.5, domain = 310:330] plot({-0.5*cos(1.5*\x)}, {3*\x/360});
    \draw[->, >=Triangle, dashed, line width=1] (1.0,0.1)--(1.0,2.9) node[pos=1,right]{$\theta$};
  \end{tikzpicture}
  \end{minipage}
  \begin{minipage}[t]{0.245\linewidth}
    \centering (d) \\[5pt]
    \begin{tikzpicture}
  \draw[violet, line width = 1.5, domain = 0:35] plot({0.5*cos(2.0*\x)}, {3*\x/360});
  \draw[red, line width = 1.5, domain = 55:215] plot({0.5*cos(2.0*\x)}, {3*\x/360});
  \draw[violet, line width = 1.5, domain = 235:360] plot({0.5*cos(2.0*\x)}, {3*\x/360});
  \draw[blue, line width = 1.5, domain = 0:125] plot({-0.5*cos(2.0*\x)}, {3*\x/360});
  \draw[green, line width = 1.5, domain = 145:305] plot({-0.5*cos(2.0*\x)}, {3*\x/360});
  \draw[blue, line width = 1.5, domain = 325:360] plot({-0.5*cos(2.0*\x)}, {3*\x/360});

  \draw[->, >=Triangle, blue, line width = 1.5, domain = 65:75] plot({-0.5*cos(2.0*\x)}, {3*\x/360});
  \draw[->, >=Triangle, red, line width = 1.5, domain = 155:165] plot({0.5*cos(2.0*\x)}, {3*\x/360});
  \draw[->, >=Triangle, green, line width = 1.5, domain = 245:255] plot({-0.5*cos(2.0*\x)}, {3*\x/360});
  \draw[->, >=Triangle, violet, line width = 1.5, domain = 335:345] plot({0.5*cos(2.0*\x)}, {3*\x/360});
  \draw[->, >=Triangle, dashed, line width=1] (1.0,0.1)--(1.0,2.9) node[pos=1,right]{$\theta$};
\end{tikzpicture}
  \end{minipage}
  \begin{minipage}[t]{0.245\linewidth}
    \centering (f) \\[5pt]
    \begin{tikzpicture}
  \draw[violet, line width = 1.5, domain = 0:35] plot({0.5*cos(2.0*\x)}, {3*\x/360});
  \draw[red, line width = 1.5, domain = 55:215] plot({0.5*cos(2.0*\x)}, {3*\x/360});
  \draw[violet, line width = 1.5, domain = 235:360] plot({0.5*cos(2.0*\x)}, {3*\x/360});
  \draw[blue, line width = 1.5, domain = 0:125] plot({-0.5*cos(2.0*\x)}, {3*\x/360});
  \draw[green, line width = 1.5, domain = 145:305] plot({-0.5*cos(2.0*\x)}, {3*\x/360});
  \draw[blue, line width = 1.5, domain = 325:360] plot({-0.5*cos(2.0*\x)}, {3*\x/360});

  \draw[->, >=Triangle, blue, line width = 1.5, domain = 65:75] plot({-0.5*cos(2.0*\x)}, {3*\x/360});
  \draw[->, >=Triangle, red, line width = 1.5, domain = 155:145] plot({0.5*cos(2.0*\x)}, {3*\x/360});
  \draw[->, >=Triangle, green, line width = 1.5, domain = 245:255] plot({-0.5*cos(2.0*\x)}, {3*\x/360});
  \draw[->, >=Triangle, violet, line width = 1.5, domain = 335:325] plot({0.5*cos(2.0*\x)}, {3*\x/360});
  \draw[->, >=Triangle, dashed, line width=1] (1.0,0.1)--(1.0,2.9) node[pos=1,right]{$\theta$};
\end{tikzpicture}
  \end{minipage}
  \begin{minipage}[t]{0.245\linewidth}
    \centering (h) \\[5pt]
    \begin{tikzpicture}
  \draw[violet, line width = 1.5, domain = 0:35] plot({0.375+0.375*cos(2.0*\x)}, {3*\x/360});
  \draw[green, line width = 1.5, domain = 55:90] plot({0.375+0.375*cos(2.0*\x)}, {3*\x/360});
  \draw[red, line width = 1.5, domain = 0:90] plot({0.375-0.375*cos(2.0*\x)}, {3*\x/360});
  \draw[blue, line width = 1.5, domain = 0:90] plot({1.5}, {3*\x/360});

  \draw[green, line width = 1.5, domain = 90:180] plot({0}, {3*\x/360});
  \draw[red, line width = 1.5, domain = 90:125] plot({1.125+0.375*cos(2.0*\x)}, {3*\x/360});
  \draw[violet, line width = 1.5, domain = 145:180] plot({1.125+0.375*cos(2.0*\x)}, {3*\x/360});
  \draw[blue, line width = 1.5, domain = 90:180] plot({1.125-0.375*cos(2.0*\x)}, {3*\x/360});

  \draw[blue, line width = 1.5, domain = 180:215] plot({0.375+0.375*cos(2.0*\x)}, {3*\x/360});
  \draw[red, line width = 1.5, domain = 235:270] plot({0.375+0.375*cos(2.0*\x)}, {3*\x/360});
  \draw[green, line width = 1.5, domain = 180:270] plot({0.375-0.375*cos(2.0*\x)}, {3*\x/360});
  \draw[violet, line width = 1.5, domain = 180:270] plot({1.5}, {3*\x/360});

  \draw[red, line width = 1.5, domain = 270:360] plot({0}, {3*\x/360});
  \draw[green, line width = 1.5, domain = 270:305] plot({1.125+0.375*cos(2.0*\x)}, {3*\x/360});
  \draw[blue, line width = 1.5, domain = 325:360] plot({1.125+0.375*cos(2.0*\x)}, {3*\x/360});
  \draw[violet, line width = 1.5, domain = 270:360] plot({1.125-0.375*cos(2.0*\x)}, {3*\x/360});

  \draw[blue, ->, >=Triangle, line width = 1.5, domain = 35:55] plot({1.5}, {3*\x/360});
  \draw[green, ->, >=Triangle, line width = 1.5, domain = 125:145] plot({0}, {3*\x/360});
  \draw[violet, ->, >=Triangle, line width = 1.5, domain = 215:235] plot({1.5}, {3*\x/360});
  \draw[red, ->, >=Triangle, line width = 1.5, domain = 305:325] plot({0}, {3*\x/360});
  \draw[->, >=Triangle, dashed, line width=1] (2.0,0.1)--(2.0,2.9) node[pos=1,right]{$\theta$};
\end{tikzpicture}
  \end{minipage}
  \caption{
    \label{fig:trefoil}
    Examples of (a) a trefoil vortex knot,
    (c) a coherently oriented vortex link,
    (e) an incoherently oriented vortex link, and
    (g) a figure-eight vortex knot.
    The knots and links in panels (a), (c), (e), and (g) are constructed by connecting the upper and lower parts of the vortex braids shown in panels (b), (d), (f), and (h), respectively.
    The orientations of the hydrodynamic circulations for each vortex are indicated by arrows.
    Different arcs are represented in different colors.
  }
\end{figure*}
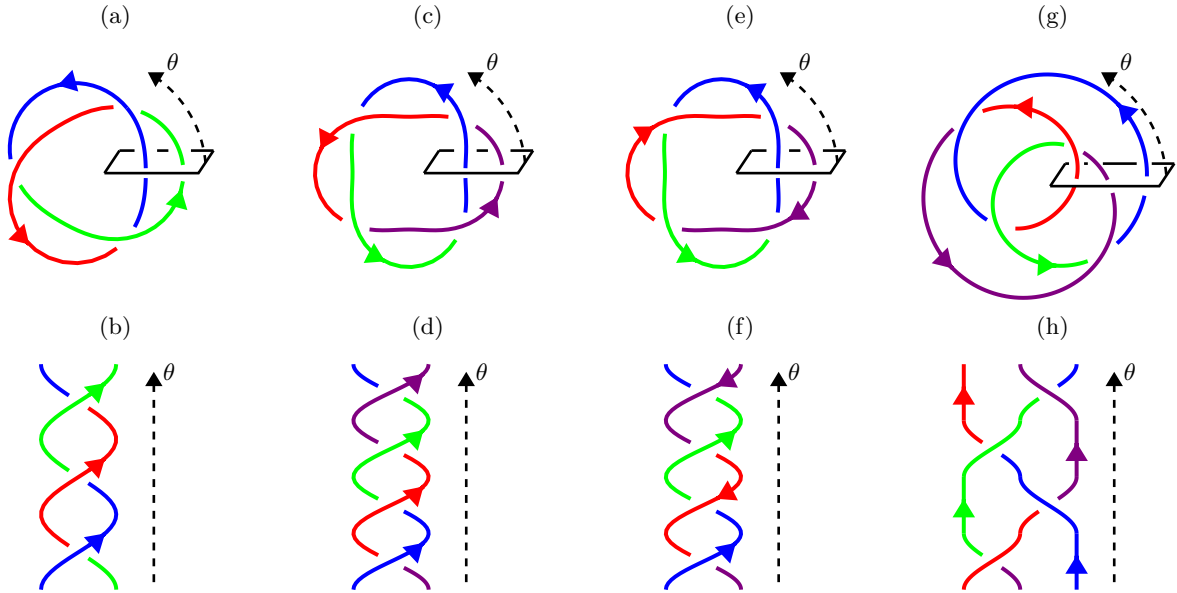
We examine vortex knots and links using a mathematical framework recently developed in \cite{NKTK24}. According to Alexander’s theorem, any knot or link can be constructed as the closure of a braid (see Appendix \ref{sec:math} in Methods). A simple example of a trefoil vortex knot is shown in Fig.~\ref{fig:trefoil}(a), which can be constructed from the circular closure of a vortex braid, as depicted in Fig.~\ref{fig:trefoil}(b). The dynamics of this vortex knot are inferred from the vortex braid’s behavior; a trefoil knot containing only one positive or negative braid will rotate in a circular motion, with the direction $\theta$ determined by the braiding in Fig.~\ref{fig:braid_xyz}. Another example, a $(2,4)$-torus link, is illustrated in Figs.~\ref{fig:trefoil}(c) and \ref{fig:trefoil}(d). Both the trefoil and the torus link are divided into several arcs, disconnected by crossing points.

The stability of vortex knots and links can be classified into two categories: (A) dynamical stability and (B) topological stability.

(A) Dynamical stability:
The dynamical stability of a vortex knot or link, when represented as a circular vortex braid, depends on (i) the orientation of the hydrodynamic circulation and (ii) the positivity of the braid. This is because the dynamics of the braids in Fig.~\ref{fig:braid_xyz} are consistent with the dynamics of point vortices in Fig.~\ref{fig:braid_xyt}. For example, a vortex knot or link will be unstable if its orientations are not aligned. Two crossing points with opposite braidings will move in opposing directions, eventually colliding. In such cases, the vortex knot or link will unravel. Fig.~\ref{fig:trefoil}(e) shows an incoherently oriented $(2,4)$-torus link, with the corresponding braid shown in Fig.~\ref{fig:trefoil}(f) containing vortices with opposite hydrodynamic circulations.

The orientation of vortex knots and links is determined by the hydrodynamic circulation of the vortex lines forming the braid. If the circulations of the vortex lines are consistent, the vortices can smoothly follow the braid and rotate around the knot or link without destabilizing. However, if the vortex lines have zero or opposite hydrodynamic circulations, the vortex knot or link is expected to be unstable. Thus, dynamical stability is characterized by (coherently) {\bf oriented positive braid links (OPB)}, which are closures of oriented positive braids.

As a reference, closed relativistic cosmic strings shrink over time \cite{Davis_1989} (see Supplementary Movies 14 \& 15), and knots or links of relativistic cosmic strings also shrink (see Supplementary Movies 16 \& 17).% Since liquid crystals have kinetic terms similar to relativistic theories, they also do not exhibit stable rings, knots, or links, even with topological stability (discussed below).

(B) Topological stability:
Vortex knots and links are composed of several arcs separated by crossing points. In Fig.~\ref{fig:trefoil}(a) and (c), the trefoil knot consists of three arcs, while the $(2,4)$-torus link consists of four. Each arc is classified by a topological charge corresponding to elements of the fundamental group $\pi_1(G/H)$, where $G$ and $H$ are Lie groups determined by the superfluid phase transition and $G \to H$ symmetry breaking.

If all the vortex arcs share the same topological charge, the vortex knot or link is Abelian and tends to decay into several loops through reconnections (see Supplementary Movies 3, 6, 8, 10, and 12). This behavior has been studied in scalar BECs \cite{Kleckner_2016,PhysRevE.85.036306}. However, in this study, we investigate {\it non-Abelian} vortex knots and links. If the topological charges of the arcs in a knot or link are {\bf mutually non-identical (MNI)}, then the structure is more stable. If the topological charges are also {\bf non-commutative (MNC)}, then the structure is even more robust. For an MNC vortex knot, the fundamental group $\pi_1(G/H)$ must be non-Abelian \cite{Poenaru1977,Mermin:1979zz}, and to further avoid reconnections, the fundamental group should not be {\bf nilpotent} \cite{Brekke_1992} (see Appendix~\ref{sec:nilpotent} of Methods).

In an MNC vortex knot, when two arcs collide, they do not pass through each other but instead create a vortex bridge, preventing reconnection \cite{Poenaru1977,Mermin:1979zz,Kobayashi:2008pk}. This bridge formation requires energy, making MNC knots stable. In contrast, for MNI knots, where the elements are commutative, arcs may pass through each other during collisions. While this does not necessarily cause instability, it makes MNI knots less stable than MNC knots, resulting in metastable structures.

In summary, stable vortex knots and links must meet both (A) dynamical and (B) topological stability conditions. These stable structures are characterized as {\bf mutually non-commutative oriented positive braid links (MNCOPB)} or their mirror images. Metastable knots and links may include {\bf mutually non-identical oriented positive braid links (MNIOPB)} and OPB, where arcs of commutative elements are spaced far apart to prevent reconnections. In the following sections, we introduce a unique physical system that can support these stable and metastable vortex knots and links.

%%%%%%%%%%%%%%%%%%%%
\section*{Stable knots and links in Bose-Einstein condensates}

We consider physical hydrodynamic systems capable of producing stable quantum vortex knots and links. The simplest such systems, which admit quantum vortices, are superfluid helium and scalar atomic BECs, where superfluidity results from the breaking of the $U(1)$ symmetry as 
$U(1) \to 1$. In these systems, vortices are characterized by topological charges classified by the fundamental group $\pi_1(U(1)/1) \cong \mathbb{Z}$. While vortex knots and links have been extensively studied in these systems, they are inherently unstable and tend to break into separate vortex loops \cite{PhysRevE.85.036306,Kleckner_2016} because the fundamental group $\mathbb{Z}$ is Abelian. In these cases, all vortex arcs of a knot or link have identical unit components and crossing points lead to reconnections that destabilize the knot or link.

The simplest experimentally accessible systems that host non-Abelian vortices, which are classified by non-Abelian fundamental groups 
$\pi_1(G/H)$, are spin-2 spinor BECs. These systems exhibit two distinct phases, defined by spin-dependent coupling constants, which give rise to non-Abelian vortices \cite{Semenoff:2006vv,Kobayashi:2008pk,Borgh:2016cco,Mawson:2018klj}, see Ref.~\cite{Kawaguchi:2012ii} as a review. In the {\bf cyclic phase}, the relevant fundamental group is $\pi_1[(U(1) \times \SO(3)) / T]\cong \mathbb{Z} \times_h T^\ast$, and in the {\bf $D_4$-nematic phase}, it is $\pi_1[(U(1) \times \SO(3)) / D_4] \cong \mathbb{Z} \times_h D_4^\ast$, where $T^\ast$ and $D_4^\ast$ are the universal covering groups of the tetrahedral group $T$ and the fourth dihedral group $D_4$, respectively. The $h$-product $\times_h$ is defined in \cite{Kobayashi:2011xb}.

Vortices in these phases are characterized by topological charges given by pairs of elements $(k, g)$, where $k$ belongs to the $U(1)$ subgroup and $g$ belongs to the $SO(3)$ subgroup. Importantly, vortices are not classified by the fundamental group itself but rather by its conjugacy classes \cite{Mermin:1979zz}. There are seven conjugacy classes in both $\mathbb{Z} \times_h T^\ast$ and $\mathbb{Z} \times_h D_4^\ast$ \cite{Semenoff:2006vv,Mawson:2018klj,Masaki_2024}, which determine the hydrodynamic properties, such as stability and circulation (see Methods Sec.~\ref{sec:spin-2BEC}).

Vortex arcs are categorized into hydrodynamic and non-hydrodynamic classes based on their $U(1)$ phase winding. Vortices classified as hydrodynamic have nonzero (fractional) $U(1)$ phase windings, while non-hydrodynamic vortices have zero winding. As with classical hydrodynamic vortex rings described by Helmholtz ~\ref{sec:spin-2BEC}, hydrodynamic vortex rings with minimal $U(1)$ phase winding are dynamically stable and move with self-induced velocities (see Supplementary Movie 2). In contrast, non-hydrodynamic vortices are dynamically unstable and tend to shrink, similar to closed relativistic cosmic strings \cite{Davis_1989} (see Supplementary Movie 1). Therefore, we focus on hydrodynamic vortices for constructing stable vortex knots and links.

Another candidate for generating non-Abelian vortices is the $D_2$-nematic phase in spin-2 spinor BECs, with the fundamental group $\pi_1[U(1) \times \SO(3) / D_2 ] \cong \mathbb{Z} \times Q_8$, where $Q_8$ is the quaternion group. Although topologically stable vortex knots have been discussed in this phase \cite{Annala_2022}, all vortices in this phase are non-hydrodynamic and quickly decay into vortex rings (see Supplementary Movie 5). Biaxial nematic liquid crystals exhibit similar behavior.

The movement of hydrodynamic vortex rings is determined by the orientation of the vortex knots and links. Like simple hydrodynamic vortex rings, coherently oriented vortex knots and links, such as those shown in Figs.~\ref{fig:trefoil}(a), \ref{fig:trefoil}(c), and \ref{fig:trefoil}(g), propagate in a single direction. In contrast, incoherently oriented vortex knots and links, such as those shown in Fig. \ref{fig:trefoil}(e), exhibit no fixed movement direction. In such cases, the vortex loops move in opposite directions, ultimately becoming unstable (see Supplementary Movie 7).

In hydrodynamic vortex knots and links, the crossing points also move along the loops, with their movement direction determined by the sign of the crossing. For OPB knots and links (see Figs.~\ref{fig:trefoil}(a) and \ref{fig:trefoil}(c)), all crossing points move in the same direction, contributing to their stability. However, in the case of the figure-eight knot (Fig. \ref{fig:trefoil}(g)), the outer and inner crossing points move in opposite directions, eventually leading to a collision and instability.

We have classified all possible MNCOPB and MNIOPB in the cyclic and nematic phases of spin-2 BECs (see Methods Sec.~\ref{sec:classification}). Among all possibilities, we identify a unique MNCOPB-- a stable vortex trefoil knot in the cyclic phase of spin-2 BEC (Fig.~\ref{fig:main-results}(a) and Supplementary Movie 4). This is the first stable vortex knot observed since Lord Kelvin's original proposal over 150 years ago, representing the main result of this work.

In addition to the MNCOPB, we have found two classes of metastable knots and links. The first class consists of structures that maintain their shape, with no collisions or rung formations. The second class includes knots and links that do not retain their shape, as rung formation occurs during collisions. In these cases, stability depends on the non-commutativity of the created rungs.

Examples of first class
are given by repeating the braids for MNCOPBs and taking their closures, namely powers of the braid 
corresponding to the MNCOPB:
the $(2,6)$-torus link (Fig.~\ref{fig:main-results}(b) and Supplementary Movie 9) and the $(2,9)$-torus knot (Fig.~\ref{fig:main-results}(c) and Supplementary Movie 11) in the cyclic phase. 
These can be generalized to $(2,3n)$-torus links or knots, where $n \in \mathbb{Z}$ and $n \neq 0$. These structures remain stable even with multiple arcs corresponding to the same topological charge. This contrasts with the $(3,3)$-torus link in the cyclic phase, which consists of three vortex loops with identical topological charges and breaks into three vortex rings (see Supplementary Movie 13).

In the second class, examples include two types of $(2,4)$-torus links in the nematic phase. Both are MNIOPBs consisting of two linked vortices, with one set belonging to the same conjugacy class (Supplementary Movie 18) and the other to different classes (Fig.~\ref{fig:main-results}(d) and Supplementary Movie 19). In both cases, the vortices collide, forming four rungs. In the first case, the four rungs have identical topological charges, making the structure unstable, and eventually splitting into two vortex rings. In the second case, the rungs consist of well-separated pairs with the same topological charge, making the structure metastable.

\bigskip
Our findings demonstrate how to tie knots and links of non-Abelian quantum vortices in systems with internal degrees of freedom, avoiding reconnections through non-commutativity. These results have potential applications in neutron stars, where the nematic phase may exist in neutron $P$-wave superfluids \cite{Chamel2017,Haskell:2017lkl,Sedrakian:2018ydt}. The existence of stable vortex knots or links in such environments could have significant implications for neutron star dynamics and evolution. 
Moreover, our mathematical framework for strongly protected knotted and linked structures may have broader applications in systems that prohibit reconnections. These include quantum computation, where the braiding of non-Abelian quantum vortices corresponds to computations \cite{Mawson:2018klj}, non-Abelian extensions of quantum turbulence (which may exhibit scaling laws distinct from the conventional Kolmogorov law for fluids and superfluids), and biological systems such as DNA, RNA, and proteins, where reconnections are prohibited \cite{Wasserman}.

%%%%%%%%%%%%%%%%%%%%%%%%%
\section*{Acknowledgments}
%The authors would like to thank .... for useful comments. 
This work is supported by 
the World Premier International Research Center Initiative (WPI) program ``Sustainability with Knotted Chiral Meta Matter (SKCM$^2$)'' 
 in the Ministry of
Education, Culture, Sports, Science (MEXT), Japan, 
%at Hiroshima University
and 
by 

Japan Society for the Promotion of Science 
(JSPS) Grant-in-Aid for Scientific Research KAKENHI Grant No.~JP19KK0066 (M.~K.), JP20K03765 (M.~K.), JP24K00593 (M.~K.), JP20K14317 (Y.~N.), JP23K12974 (Y.~N.), JP20K03588 (Y.~K.), JP23K20791 (Y.~K.), JP23H05437 (Y.~K.), JP24K06744 (Y.~K.), JP22H01221 
(M.~K. and M.~N.). 
%Y. K. is also supported  JST
%CREST Grant Number JPMJCR17J4.

%%%%%%%%%%%%%%%%%%%%%%%%%%%%%%%%%

%%%%%%%%%%%%%%%%%%
% \bibliographystyle{ieeetr}
\bibliography{reference}

%%%%%%%%%%%%%%%%%%%%%%%%%%%%%%%%
\newpage
\begin{center}
    {\bf \LARGE Methods}
\end{center}

\section{Vortices in Spinor BECs}\label{sec:spin-2BEC}
Here, we summarize the properties of spin-2 spinor BECs~\cite{Kawaguchi:2012ii}. The order parameter of spin-2 BECs consists of five complex components, represented by a traceless symmetric $3 \times 3$ tensor $A$, with complex components transforming under the symmetry group $G = U(1) \times \SO(3)$ as:
\begin{align}
    A \mapsto e^{i\theta} g A g^T, \quad e^{i\theta} \in U(1), \quad g \in \SO(3). \label{eq:U1SO3}
\end{align}
The effective low-energy Hamiltonian density $h$ of spin-2 spinor BECs with particle mass $m$ can be expressed as:
\begin{align} \begin{aligned}
    h &= \frac{\hbar^2}{6m} \sum_{i,j=1,2,3} (\nabla A_{ij}^\ast) \cdot (\nabla A_{ij}) \\ &+ \frac{1}{2} \left(c_0 \rho^2 + c_1 \boldsymbol{S}^2 + c_2 |\Psi_{20}|^2\right),
\end{aligned} \end{align}
where $\rho$, $\boldsymbol{S}^2$, and $\Psi_{20}$ are the density, the spin density, and the singlet-pair amplitude, respectively, defined as:
\begin{align} \begin{aligned}
    & \rho = \frac{1}{3} \left( \operatorname{tr} A^\ast A \right), \quad \Psi_{20} = \frac{1}{3} \left( \operatorname{tr} A^2 \right), \\ & \boldsymbol{S}^2 = \frac{4}{9} \left(\operatorname{tr} A^\ast A\right)^2 + \frac{2}{9} \left|\operatorname{tr}A^2\right|^2 - \frac{4}{3}\left(\operatorname{tr}A^{\ast 2} A^2 \right).
\end{aligned} \end{align}
The Hamiltonian density $h$ is invariant under the transformation defined in Eq.\eqref{eq:U1SO3}. The phases of condensates with the total angular momentum of two are classified as nematic for $c_2 < 0$ and $c_2 < 4 c_1$, cyclic for $c_2 > 0$, and ferromagnetic for $c_1 < 0$ and $c_2 > 4 c_1$ \cite{Mermin:1974zz}. All these phases are theoretically possible in spin-2 BECs. Experimentally, spin-2 BECs have been realized using $^{87}$Rb atoms, where the phase lies near the boundary between the cyclic and ferromagnetic phases (see, e.g., Ref.\cite{PhysRevA.80.042704}).

In this paper, we focus on the nematic and cyclic phases, which host non-Abelian vortices.

The dynamics of vortices in spinor BECs, in terms of the traceless symmetric tensor $A$, is described by the Gross-Pitaevskii equation:
    \begin{align} i \frac{\partial A_{ij}}{\partial t} = \frac{\delta h}{\delta A_{ij}^\ast}.
\end{align}
The Gross-Pitaevskii equation is non-relativistic, but it can be compared to the relativistic equation used in cosmology:
\begin{align}
    \frac{\partial^2 A_{ij}}{\partial t^2} = - \frac{\delta h}{\delta A_{ij}^\ast},
\end{align}
which admits relativistic cosmic strings.

%%%%%%%%%%%

\subsection{Cyclic phase in spin-2 BEC}

The order parameter of the cyclic phase is:
\begin{align}
    A \sim {\rm diag} (1,e^{2\pi i/3},e^{4\pi i/3}), \label{eq:cyclic}
\end{align}
which leads to symmetry breaking down to the tetrahedral group $H \simeq T$, and the order parameter manifold is given by:
\begin{align}
    \frac{G}{H} = \frac{U(1) \times \SO(3)}{T} \simeq \frac{U(1) \times SU(2)}{T^\ast},
\end{align}
with $T^\ast$, and a non-Abelian fundamental group:
\begin{align}
    \pi_1 \left( \frac{U(1) \times SU(2)}{T^{\ast}}\right) \cong \mathbb{Z} \times_h T^{\ast},
    \label{conj-1}
\end{align}
where $\times_h$ is defined in Ref.~\cite{Kobayashi:2011xb}. This group consists of the following twenty-four types of elements:
\begin{align}
    \left\{\begin{array}{l}
        \displaystyle
        (N, \pm \bm{1}_2), \\
        \displaystyle
        (N, \pm i \sigma_x), (N,\pm i  \sigma_y), 
        (N, \pm i \sigma_z), \\
        \displaystyle
        \left(N + \frac{1}{3}, \pm C_3\right), 
        \left(N + \frac{1}{3}, \pm i \sigma_x C_3 \right), \\
        \displaystyle
        \left(N + \frac{1}{3}, \pm i \sigma_y C_3 \right),
        \left(N + \frac{1}{3}, \pm i \sigma_z C_3 \right), \\
        \displaystyle
        \left(N - \frac{1}{ 3}, \pm C_3^2 \right),
        \left(N - \frac{1}{ 3}, \pm i \sigma_x C_3^2 \right), \\ 
        \displaystyle
        \left(N - \frac{1}{ 3}, \pm i \sigma_y C_3^2\right), 
        \left(N - \frac{1}{ 3}, \pm i \sigma_z C_3^2\right)
    \end{array}\right\},
    \label{cyclic-elements}
\end{align}
where the first and second elements of a pair $(\cdot,\cdot)$ denote the circulation $\kappa$ (i.e., the $U(1)$ winding number) and an $SU(2)$ element, respectively, with $N \in \mathbb{Z}$. The matrix $C_3 \equiv (1/2) (\bm{1}_2 + i \sigma_x + i \sigma_y + i \sigma_z)$ satisfies $(C_3)^3 = - \bm{1}_2$.

\begin{table*}[htb]
    \centering
    \begin{tabular}{c|ccc}
        Phase & & Conjugacy class & $N$ for stable vortex \\ \hline
        \multirow{18}{*}{Cyclic} & (I) & $\{(N, {\bm 1}_2)\}$ & unstable \\
        & (II) & $\{(N, - {\bm 1}_2)\}$ & unstable  \\
        & (III) & $\{(N, \pm i \sigma_x), (N, \pm i \sigma_y), (N, \pm i \sigma_z)\}$ & $N=0$ (non-hydrodynamic) \\
        & (IV) & $\left\{\begin{array}{l}
        \displaystyle \left(N+\frac{1}{3}, C_3\right), \left(N+\frac{1}{3}, -i \sigma_x C_3\right) \\
        \displaystyle \left(N+\frac{1}{3}, -i \sigma_y C_3 \right), \left(N+\frac{1}{3}, -i \sigma_z C_3 \right)
        \end{array}\right\}$ & $N = 0$ (hydrodynamic) \\
        & (V) & $\left\{\begin{array}{l}
        \displaystyle \left(N+\frac{1}{3}, (C_3)^{-2} \right), \left(N+\frac{1}{3}, (i \sigma_x C_3)^{-2} \right) \\
        \displaystyle \left(N+\frac{1}{3}, (i \sigma_y C_3)^{-2} \right), \left(N+\frac{1}{3}, (i \sigma_z C_3)^{-2}\right)
        \end{array}\right\}$ & unstable  \\
        & (VI) & $\left\{\begin{array}{l}
        \displaystyle \left(N-\frac{1}{3}, (C_3)^2\right), \left(N-\frac{1}{3}, (i \sigma_x C_3)^2 \right) \\
        \displaystyle \left(N-\frac{1}{3}, (i \sigma_y C_3)^2 \right), \left(N-\frac{1}{3}, (i \sigma_z C_3)^2 \right)
        \end{array}\right\}$ & unstable  \\
        & (VII) & $\left\{\begin{array}{l}
        \displaystyle \left(N-\frac{1}{3}, (C_3)^{-1} \right), \left(N-\frac{1}{3}, (-i \sigma_x C_3)^{-1}\right) \\
        \displaystyle \left(N-\frac{1}{3}, (-i \sigma_y C_3)^{-1}\right), \left(N-\frac{1}{3}, (-i \sigma_z C_3)^{-1}\right)
        \end{array}\right\}$ & $N = 0$ (hydrodynamic) \\ \hline
        \multirow{12}{*}{$D_4$-BN} & (I) & $\{(N, {\bm 1}_2)\}$ & unstable \\
        & (II) & $\{(N, - {\bm 1}_2)\}$ & unstable \\
        & (III) & $\{(N, \pm i \sigma_x), (N, \pm i \sigma_y)\}$ & $N=0$ (non-hydrodynamic) \\
        & (IV) & $\{(N, \pm i \sigma_z)\}$ & unstable \\
        & (V) & $\displaystyle \left\{
        \left(N+\frac{1}{2}, C_4\right), \left(N+\frac{1}{2}, (C_4)^{-1} \right)
        \right\}$ & $N = 0$, $-1$ (hydrodynamic) \\
        & (VI) & $\displaystyle \left\{\
        \left(N+\frac{1}{2}, (C_4)^3\right), \left(N+\frac{1}{2}, (C_4)^{-3} \right)
        \right\}$ & unstable \\
        & (VII) & $\left\{\begin{array}{l}
        \displaystyle \left(N+\frac{1}{2}, i \sigma_x C_4\right), \left(N+\frac{1}{2}, (i \sigma_x C_4)^{-1} \right) \\
        \displaystyle \left(N+\frac{1}{2}, i \sigma_x (C_4)^{-1} \right), \left(N+\frac{1}{2}, ( i \sigma_x (C_4)^{-1})^{-1} \right)
        \end{array}\right\}$ & $N=0$, $-1$ (hydrodynamic) \\
    \end{tabular}
    \caption{
        \label{table-spin2}
        Conjugacy classes of fundamental groups for the order-parameter manifold in cyclic (upper) and $D_4$-BN (lower) phases. All components of the fundamental groups are given in Eqs.~\eqref{cyclic-elements} and \eqref{conj-5} for the cyclic and $D_4$-BN phases, respectively. In the table, we use the relations $(C_3)^{-1} = - (C_3)^2$ and $(C_4)^{-1} = -(C_4)^3$.
    }
\end{table*}

The conjugacy classes of Eq.~(\ref{cyclic-elements}) are composed of the seven elements for each $N$ \cite{Semenoff:2006vv}, as summarized in Table \ref{table-spin2}. These elements describe:
(I) integer vortices ($N=0$ corresponds to the vacuum),
(II) spin vortices with $2\pi$ rotations,
(III) spin vortices with $\pi$ rotations around the $x$, $y$, or $z$ axes, and
(IV)--(VII) non-Abelian $\pm 1/3$ ($2/3$)-quantum vortices.
The stability of $\pm 1/3$ ($2/3$)-quantum vortices can be summarized as follows. Vortices (IV) with $N=0$ ($\kappa=1/3$) and (VII) with $N=0$ ($\kappa=-1/3$) are stable and can be considered elementary. These two types are anti-vortices to each other: (IV) = (VII)$^{-1}$. Vortices (IV) with $N=-1$ ($\kappa=-2/3$) and (VII) with $N=1$ ($\kappa=2/3$) are marginally stable
%\footnote{They remain stable unless a vortex with non-commutative topological charge approaches them. In such cases, they may split into elementary vortices.
%\red{(NO NEED since the same kind of sentence outside the footnote?)}}
, but they can split into two elementary vortices when approached by a vortex carrying a non-commutative topological charge. The remaining vortices are unstable and decay into elementary vortices ($\kappa = \pm 1/3$) or a combination of elementary vortices and marginally stable vortices.

In simulations, we observe that marginally stable vortices ($\kappa = \pm 2/3$) are prone to decay into a pair of elementary vortices when near a vortex with non-commutative topological charge (see Supplementary Movie 21). However, such decay does not occur when near a vortex with a commutative topological charge (see Supplementary Movie 22). Thus, marginally stable vortices ($\kappa = \pm 2/3$) are unfavorable for forming stable MNCOPB or MNIOPB, and we do not consider them further in this paper.

%%%%%%%%%%%%%%%%%%%%%%%%%
\subsection{$D_4$-BN phases in spin-2 BEC}

The nematic phase consists of three degenerate phases: the UN, $D_2$-, and $D_4$-BN phases. Quantum or thermal effects can cause the system to select one of these phases \cite{Uchino:2010pf, Kobayashi:2011xb}. Here, we focus on the $D_4$-BN phase.\footnote{This $D_4$-BN phase is also found in $P$-wave neutron superfluids in neutron stars, where it supports singly quantized vortices \cite{Sauls:1978lna, Muzikar:1980as, Sauls:1982ie, Masuda:2015jka, Masaki:2019rsz} and half-quantum vortices \cite{Masuda:2016vak, Masaki:2021hmk, Kobayashi:2022moc, Kobayashi:2022dae}.}

The order parameter for the $D_4$-BN phase is:
\begin{align}
    A \sim {\rm diag} (1,-1,0),
\end{align}
and the corresponding order parameter manifold is:
\begin{eqnarray}
    \frac{G}{H} = \frac{U(1) \times \SO(3)}{D_4} \simeq \frac{U(1) \times SU(2)}{D_4^{\ast}},
\end{eqnarray}
with $D_4^{\ast}$ being a covering group. The non-Abelian fundamental group is:
\begin{align}
    \pi_1 \left( \frac{U(1) \times SU(2)}{D_4^{\ast}}\right) \cong \mathbb{Z} \times_h D_4^{\ast},
    \label{conj-4}
\end{align}
and contains sixteen elements for each $N$:
\begin{align}
    \left\{ \begin{array}{l} \displaystyle (N, \pm \bm{1}_2), \\ \displaystyle (N, \pm i \sigma_x), (N, \pm i \sigma_y), \\ \displaystyle (N, \pm i \sigma_z), \\ \displaystyle \left(N + \frac{1}{2}, \pm C_4\right), \left(N + \frac{1}{2}, \pm i \sigma_x C_4\right), \\ \displaystyle \left(N + \frac{1}{2}, \pm C_4^{-1} \right), \left(N + \frac{1}{2}, \pm i \sigma_x (C_4)^{-1} \right) \end{array} \right\},
    \label{conj-5}
\end{align}
where the notation follows Eq.~(\ref{cyclic-elements}). Here, $C_4 \equiv e^{i \frac{\pi}{4}\sigma_z} = (1/\sqrt{2})(\bm{1}_2 + i \sigma_z)$ satisfies $C_4^4 = -\bm{1}_2$.

The conjugacy classes of Eq.~(\ref{conj-5}) consist of seven elements for each $N$, as shown in Table~\ref{table-spin2}. These elements describe:
(I) integer vortices ($N=0$ corresponds to the vacuum),
(II) spin vortices with $2\pi$ rotations,
(III) and (IV) spin vortices with $\pi$ rotations around the $x$, $y$, and $z$ axes, and
(V)--(VII) non-Abelian half-quantum vortices (HQVs).
The stability of the HQVs is as follows. Vortices (V) and (VII) with $N=0$ ($\kappa = 1/2$) and $N=-1$ ($\kappa = -1/2$) are stable and can be considered elementary. The other vortices are unstable, and there are no marginally stable vortices.

%%%
\section{Classification of vortex knots and links}\label{sec:classification}

Here, we exhaust all possible mutually non-commutative oriented positive braids (MNCOPBs) and mutually non-identical oriented positive braids (MNIOPBs) in spin-2 BECs. By definition, there are only finitely many MNCOPBs and MNIOPBs (up to circulation). A complete list of these can be obtained by simple enumeration.

\begin{proposition}
In cyclic spin-2 BECs, there are exactly eight MNCOPBs whose colors correspond to the elements of conjugacy classes (IV) and (VII), for up to two-component links. These braids are trefoil knots, with specific colors as shown in the left panel of Fig.~\ref{fig:mnob_cyclic}. The colorings $\begin{pmatrix} a & b  \\ c & \end{pmatrix}$ correspond to one of the following:
\begin{enumerate}
    \item 
    $\begin{pmatrix}  \left(N+\frac{1}{3}, -i \sigma_z C_3 \right) &
      \left(N+\frac{1}{3}, -i \sigma_y C_3 \right)\\
      \left(N+\frac{1}{3}, -i \sigma_x C_3\right) & 
       \end{pmatrix}$, 
    \item 
    $\begin{pmatrix} \left(N+\frac{1}{3}, -i \sigma_y C_3 \right) &  \left(N+\frac{1}{3}, -i \sigma_z C_3 \right)  \\
     \left(N+\frac{1}{3}, C_3\right)  &
       \end{pmatrix}$, 
    \item 
    $\begin{pmatrix} \left(N+\frac{1}{3}, C_3\right) & 
      \left(N+\frac{1}{3}, -i \sigma_x C_3\right) \\
       \left(N+\frac{1}{3}, -i \sigma_y C_3 \right) &
       \end{pmatrix}$, 
    \item 
    $\begin{pmatrix} \left(N+\frac{1}{3}, C_3\right) & 
     \left(N+\frac{1}{3}, -i \sigma_z C_3 \right) \\ 
      \left(N+\frac{1}{3}, -i \sigma_x C_3\right) &
       \end{pmatrix}$, 
    \item 
    $\begin{pmatrix} \left(N-\frac{1}{3}, (-i \sigma_z C_3)^{-1}\right) &
      \left(N-\frac{1}{3}, (-i \sigma_y C_3)^{-1}\right) \\
      \left(N-\frac{1}{3}, (-i \sigma_x C_3)^{-1}\right) &
       \end{pmatrix}$, 
    \item 
    $\begin{pmatrix} \left(N-\frac{1}{3}, (C_3)^{-1} \right) & 
      \left(N-\frac{1}{3}, (-i \sigma_x C_3)^{-1}\right) \\
      \left(N-\frac{1}{3}, (-i \sigma_y C_3)^{-1}\right) &
       \end{pmatrix}$, 
    \item 
    $\begin{pmatrix} \left(N-\frac{1}{3}, (-i \sigma_z C_3)^{-1}\right) &
      \left(N-\frac{1}{3}, (C_3)^{-1} \right) \\ 
      \left(N-\frac{1}{3}, (-i \sigma_y C_3)^{-1}\right) &
       \end{pmatrix}$, 
    \item 
    $\begin{pmatrix} \left(N-\frac{1}{3}, (-i \sigma_x C_3)^{-1}\right) &
      \left(N-\frac{1}{3}, (C_3)^{-1} \right) \\
      \left(N-\frac{1}{3}, (-i \sigma_z C_3)^{-1}\right) &
       \end{pmatrix}$.
\end{enumerate}
The four colorings 1--4 (and similarly 5--8) are equivalent to each other up to simultaneous conjugation.

Additionally, there are some incoherently oriented MNCOPBs. These include the $(2,4)$-torus link, whose specific colors are as shown in the right panel of Fig.~\ref{fig:mnob_cyclic}. The colorings $\begin{pmatrix} d & e \\ f & g \end{pmatrix}$ are one of the following:
\begin{enumerate}
    \item 
    $\begin{pmatrix}  \left(N+\frac{1}{3}, C_3\right) &
      \left(N+\frac{1}{3}, -i \sigma_x C_3\right)\\  
      \left(N+\frac{1}{3}, -i \sigma_y C_3\right)&
      \left(N+\frac{1}{3}, -i \sigma_y C_3\right)
       \end{pmatrix}$, 
    \item 
    $\begin{pmatrix} \left(N+\frac{1}{3}, C_3\right)&
      \left(N+\frac{1}{3}, -i \sigma_y C_3\right)\\
      \left(N+\frac{1}{3}, -i \sigma_x C_3 \right)&
      \left(N+\frac{1}{3}, -i \sigma_z C_3\right)
       \end{pmatrix}$, 
    \item 
    $\begin{pmatrix} \left(N+\frac{1}{3}, C_3\right)& 
      \left(N+\frac{1}{3}, -i \sigma_z C_3\right)\\  
      \left(N+\frac{1}{3}, -i \sigma_z C_3 \right)&
      \left(N+\frac{1}{3}, -i \sigma_x C_3\right)
       \end{pmatrix}$, 
    \item 
    $\begin{pmatrix} \left(N+\frac{1}{3}, -i \sigma_x C_3\right)&
      \left(N+\frac{1}{3}, C_3 \right)\\
      \left(N+\frac{1}{3}, -i \sigma_y C_3 \right)& 
      \left(N+\frac{1}{3},  -i \sigma_y C_3 \right)
       \end{pmatrix}$, 
    \item 
    $\begin{pmatrix} \left(N+\frac{1}{3}, -i \sigma_x C_3\right)&
      \left(N+\frac{1}{3}, -i \sigma_z C_3\right)\\
      \left(N+\frac{1}{3}, C_3 \right)&
      \left(N+\frac{1}{3},  -i \sigma_y C_3 \right)
       \end{pmatrix}$, 
    \item 
    $\begin{pmatrix} \left(N+\frac{1}{3}, -i \sigma_z C_3 \right)&
      \left(N+\frac{1}{3}, C_3 \right)\\
      \left(N+\frac{1}{3}, -i \sigma_y C_3 \right)&
      \left(N+\frac{1}{3}, -i \sigma_x C_3 \right)
       \end{pmatrix}$. 
\end{enumerate}
The three colorings 1--3 (and similarly 4--6) are equivalent to each other up to simultaneous conjugation.
\end{proposition}

\begin{figure}[htb]
  \centering
  \begin{picture}(400,0)(0,0)
\put(47,2){$a$}
\put(3,-65){$b$}
\put(88,-65){$c$}

\put(143,-5){$d$}
\put(143,-75){$e$}
\put(213,-75){$f$}
\put(213,-5){$g$}
\end{picture}
\includegraphics[height=2.5cm]{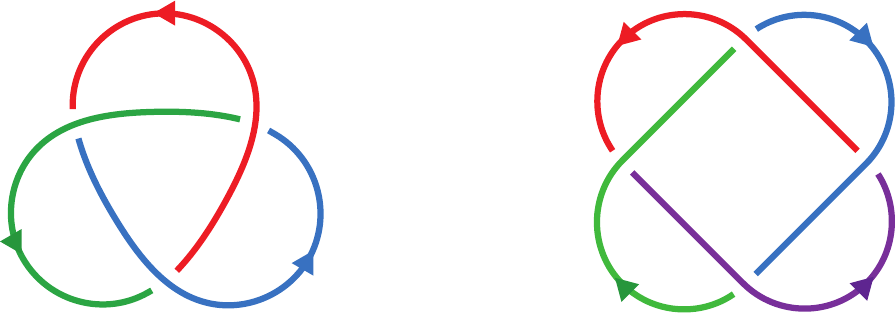}
  \caption{
  \label{fig:mnob_cyclic}
  The MNCOPBs (left) 
  and incoherently oriented MNCPB (right)}
  in cyclic spin-$2$ BEC.
\end{figure}

%%%%%%%%%%%%%%%%%%%%%%
%%%%%%%%%%%%%%%%%%%%%%

\begin{proposition}
In the $D_4$-BN phase of spin-2 BECs, there are no MNCOPBs. However, there are exactly four MNIOPBs whose colors correspond to elements of conjugacy classes (V) and (VII). Additionally, there are six MNIOPBs whose colors also correspond to elements of conjugacy classes (V) and (VII). All of these braids form a $(2,4)$-torus link, with specific colorings as shown in Fig.~\ref{fig:mnob_bn}. The colors $\begin{pmatrix} a & b \\ c & d \end{pmatrix}$ correspond to one of the following:
\begin{enumerate}
    \item 
    $\begin{pmatrix}  \left(N+\frac{1}{2}, C_4\right) & 
      \left(N+\frac{1}{2}, i \sigma_x C_4\right)\\
      \left(N+\frac{1}{2}, C_4^{-1} \right)&
      \left(N+\frac{1}{2}, i \sigma_x (C_4)^{-1} \right)
    \end{pmatrix} $, 
    \item 
    $\begin{pmatrix}  \left(N+\frac{1}{2}, C_4\right) &
    \left(N+\frac{1}{2}, (i \sigma_x C_4)^{-1} \right) \\ 
    \left(N+\frac{1}{2}, C_4^{-1} \right)&
    \left(N+\frac{1}{2}, ( i \sigma_x (C_4)^{-1})^{-1} \right) 
    \end{pmatrix} $, 
    \item 
    $\begin{pmatrix}  \left(N+\frac{1}{2}, C_4\right) &
    \left(N+\frac{1}{2}, ( i \sigma_x (C_4)^{-1})^{-1} \right) \\
    \left(N+\frac{1}{2}, C_4^{-1} \right)&
    \left(N+\frac{1}{2}, i \sigma_x C_4\right) 
    \end{pmatrix} $, 
    \item 
    $\begin{pmatrix} 
    \left(N+\frac{1}{2}, C_4\right) & 
    \left(N+\frac{1}{2}, i \sigma_x (C_4)^{-1} \right)\\
    \left(N+\frac{1}{2}, C_4^{-1} \right) &
    \left(N+\frac{1}{2}, (i \sigma_x C_4)^{-1} \right) 
    \end{pmatrix} $, 
    \item 
    $\begin{pmatrix} 
    \left(N+\frac{1}{2}, i \sigma_x C_4 \right) & 
    \left(N+\frac{1}{2}, i \sigma_x (C_4)^{-1} \right)\\
    \left(N+\frac{1}{2}, -i \sigma_x C_4 \right) &
    \left(N+\frac{1}{2}, -i \sigma_x (C_4)^{-1} \right) 
    \end{pmatrix} $, 
    \item 
    $\begin{pmatrix} 
    \left(N+\frac{1}{2}, i \sigma_x C_4 \right) & 
    \left(N+\frac{1}{2}, -i \sigma_x (C_4)^{-1} \right)\\
    \left(N+\frac{1}{2}, -i \sigma_x C_4 \right) &
    \left(N+\frac{1}{2}, i \sigma_x (C_4)^{-1} \right) 
    \end{pmatrix} $.
\end{enumerate}
Colorings 1–4 (and similarly 5–6) are equivalent to each other up to simultaneous conjugation.
\end{proposition}
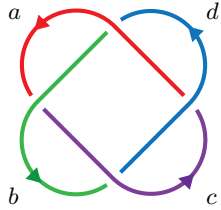
\begin{figure}[htb]
  \centering
  \begin{picture}(400,0)(0,0)
\put(73,-5){$a$}
\put(73,-75){$b$}
\put(148,-75){$c$}
\put(148,-5){$d$}
\end{picture}
\includegraphics[height=2.5cm]{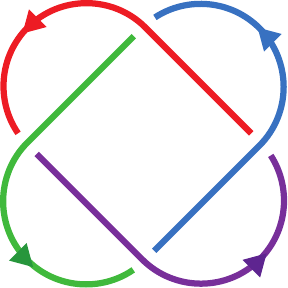}
  \caption{
  \label{fig:mnob_bn}
  The MNIOPBs
  in $D_4$-BN spin-$2$ BEC.}
\end{figure}

%%%%%%%%%%%%%%%%%%%%%%%%%%
\appendix

%%%
\section{Mathematical backgrounds}
\label{sec:math}
Unless otherwise stated, the strings of braids are oriented in the same direction.

\begin{theorem}[Alexander, cf.\ {\cite[Theorem 2.3]{Kassel_2008}}]
Every oriented link is isotopic to the closure of a braid.
\end{theorem}

A braid $\beta$ is said to be \emph{positive} if $\beta$ admits a diagram consisting solely of positive crossings. We are interested in an oriented link that is isotopic to the closure of a positive braid, referred to as a \emph{positive braid link}, while its mirror image is called a \emph{negative braid link}.

For example, the oriented Hopf link $H_+$ (resp.\ $H_-$) is a positive (resp.\ negative) braid link. A $(2,4)$-torus link with the usual orientation is a positive braid link, but when the orientation of one component is reversed, the resulting link is neither a positive nor a negative braid link. Indeed, the Jones polynomials of the link and its mirror image do not satisfy the condition given in \cite[Theorem~1]{Stoimenow_2005}.

\label{sec:nilpotent}
Let $G$ be a group, and let $H$ and $K$ be subgroups of $G$. We denote by $[H,K]$ the subgroup of $G$ generated by commutators $[x,y] = xyx^{-1}y^{-1}$ for $x \in H$ and $y \in K$. For a positive integer $n$, we define the subgroup $\Gamma_n G$ inductively by $\Gamma_1 G = G$ and $\Gamma_{n+1} G = [\Gamma_n G, G]$. A group $G$ is said to be \emph{nilpotent} if $\Gamma_n G = {1}$ for some $n$. In particular, $\Gamma_2 G = {1}$ is equivalent to $G$ being Abelian.

\begin{lemma}
\label{lem:G3G2}
Let $G$ be a group normally generated by a single element. Then, $\Gamma_n G = \Gamma_2 G$ for all $n \geq 2$.
\end{lemma}

\begin{proof}
It suffices to show that $\Gamma_3 G = \Gamma_2 G$. By assumption, there exists an element $\mu \in G$ such that the set ${g \mu g^{-1} \mid g \in G}$ generates $G$. We recall the identities $[x, yz] = [x, y] y [x, z] y^{-1}$ and $[x, y^{-1}] = y^{-1} [y, x] y$. Therefore, the proof reduces to showing that $[g \mu g^{-1}, h \mu h^{-1}] \in \Gamma_3 G$ for any $g, h \in G$. Since $[g \mu g^{-1}, h \mu h^{-1}] = [[g, \mu] \mu, [h, \mu] \mu]$ and $[\mu, \mu] = 1$, the above identities complete the proof.
\end{proof}

\begin{proposition}
\label{prop:nilpotent}
Let $G$ be a group normally generated by a single element, and let $N$ be a nilpotent group. Then, for any homomorphism $f \colon G \to N$, the image of $f$ is Abelian.
\end{proposition}

\begin{proof}
By assumption, $\Gamma_n N = {1}$ for some $n$. If $n \leq 2$, we have $\Gamma_2 N = {1}$, meaning that $N$ is Abelian. Suppose $n \geq 3$. Then, Lemma~\ref{lem:G3G2} implies that $G/\Gamma_n G = G/\Gamma_2 G$, which is Abelian. Since $\Gamma_n N = {1}$, the homomorphism $f$ factors through a homomorphism $G/\Gamma_n G \to N$, and hence the image of $f$ is Abelian.
\end{proof}

Since the fundamental group $\pi_1(S^3 \setminus K)$ of the complement of a knot in the $3$-sphere $S^3$ is normally generated by a meridian, Proposition~\ref{prop:nilpotent} yields a result mentioned in a footnote in \cite[p.~6]{Brekke_1992}.
%%%%%%

\clearpage

\onecolumngrid

\section{Simulations}

We show details for videos.
In the video, we show isosurfaces of the spin density $\boldsymbol{S}^2 = 0.06$ in blue and the singlet-pair amplitude $|A_{20}|^2 = 0.06$ in green for the cyclic phase, and the spin density  $\boldsymbol{S}^2 = 0.03$ in blue and the singlet-trio amplitude $|A_{30}|^2 = 0.03$ in green for the $D_4$ nematic phase.

\subsection{}
\begin{minipage}[t]{0.4\linewidth}
    file name: 01\_C\_1\_1\_1.mp4 \\
    phase: cyclic \\
    dynamics: non-relativistic \\
    knot type: simple ring \\
    vortex type: non-hydrodynamic
\end{minipage}
\begin{minipage}[t]{0.6\linewidth}
    vortex charges:
    $
\end{minipage}

\subsection{}
\begin{minipage}[t]{0.4\linewidth}
    file name: 02\_C\_1\_1\_4.mp4 \\
    phase: cyclic \\
    dynamics: non-relativistic \\
    knot type: simple ring \\
    vortex type: hydrodynamic
\end{minipage}
\begin{minipage}[t]{0.6\linewidth}
    vortex charges:
    $%
\end{minipage}

\subsection{}
\begin{minipage}[t]{0.4\linewidth}
    file name: 03\_C\_1\_23\_44.mp4 \\
    phase: cyclic \\
    dynamics: non-relativistic \\
    knot type: trefoil \\
    vortex type: hydrodynamic \\
    commutativity: fully identical
\end{minipage}
\begin{minipage}[t]{0.6\linewidth}
    vortex charges:
    $%
\end{minipage}

\subsection{}
\begin{minipage}[t]{0.4\linewidth}
    file name: 04\_C\_1\_23\_45.mp4 \\
    phase: cyclic \\
    dynamics: non-relativistic \\
    knot type: trefoil \\
    vortex type: hydrodynamic \\
    commutativity: MNCOPB
\end{minipage}
\begin{minipage}[t]{0.6\linewidth}
    vortex charges:
    {\color{red} $%
\end{minipage}

\subsection{}
\begin{minipage}[t]{0.4\linewidth}
    file name: 05\_C\_1\_24\_12.mp4 \\
    phase: cyclic \\
    dynamics: non-relativistic \\
    knot type: $(2,4)$-torus link \\
    vortex type: non-hydrodynamic \\
    commutativity: mutually non-identical
\end{minipage}
\begin{minipage}[t]{0.6\linewidth}
    vortex charges:
    {\color{red} $%
\end{minipage}

\subsection{}
\begin{minipage}[t]{0.4\linewidth}
    file name: 06\_C\_1\_24\_44.mp4 \\
    phase: cyclic \\
    dynamics: non-relativistic \\
    knot type: $(2,4)$-torus link \\
    vortex type: hydrodynamic \\
    commutativity: fully identical
\end{minipage}
\begin{minipage}[t]{0.6\linewidth}
    vortex charges:
    $%
\end{minipage}

\subsection{}
\begin{minipage}[t]{0.4\linewidth}
    file name: 07\_C\_1\_24\_49.mp4 \\
    phase: cyclic \\
    dynamics: non-relativistic \\
    knot type: $(2,4)$-torus link \\
    vortex type: hydrodynamic \\
    commutativity: mutually non-commutative
\end{minipage}
\begin{minipage}[t]{0.6\linewidth}
    vortex charges:
    {\color{red} $%
\end{minipage}

\subsection{}
\begin{minipage}[t]{0.4\linewidth}
    file name: 08\_C\_1\_26\_44.mp4 \\
    phase: cyclic \\
    dynamics: non-relativistic \\
    knot type: $(2,6)$-torus link \\
    vortex type: hydrodynamic \\
    commutativity: fully identical
\end{minipage}
\begin{minipage}[t]{0.6\linewidth}
    vortex charges: $%
\end{minipage}

\subsection{}
\begin{minipage}[t]{0.4\linewidth}
    file name: 09\_C\_1\_26\_45.mp4 \\
    phase: cyclic \\
    dynamics: non-relativistic \\
    knot type: $(2,6)$-torus link \\
    vortex type: hydrodynamic \\
    commutativity: non-commutative
\end{minipage}
\begin{minipage}[t]{0.6\linewidth}
    vortex charges:
    {\color{red} $%
\end{minipage}

\subsection{}
\begin{minipage}[t]{0.4\linewidth}
    file name: 10\_C\_1\_29\_44.mp4 \\
    phase: cyclic \\
    dynamics: non-relativistic \\
    knot type: $(2,9)$-torus knot \\
    vortex type: hydrodynamic \\
    commutativity: fully identical
\end{minipage}
\begin{minipage}[t]{0.6\linewidth}
    vortex charges:
    $%
\end{minipage}

\subsection{}
\begin{minipage}[t]{0.4\linewidth}
    file name: 11\_C\_1\_29\_45.mp4 \\
    phase: cyclic \\
    dynamics: non-relativistic \\
    knot type: $(2,9)$-torus knot \\
    vortex type: hydrodynamic \\
    commutativity: non-commutative \\
\end{minipage}
\begin{minipage}[t]{0.6\linewidth}
    vortex charges:
    {\color{red} $%
\end{minipage}

\subsection{}
\begin{minipage}[t]{0.4\linewidth}
file name: 12\_C\_1\_33\_444.mp4 \\
phase: cyclic \\
dynamics: non-relativistic \\
knot type: $(3,3)$-torus link \\
vortex type: hydrodynamic \\
commutativity: fully identical
\end{minipage}
\begin{minipage}[t]{0.6\linewidth}
    vortex charges:
    $%
\end{minipage}

\subsection{}
\begin{minipage}[t]{0.4\linewidth}
    file name: 13\_C\_1\_33\_456.mp4 \\
    phase: cyclic \\
    dynamics: non-relativistic \\
    knot type: $(3,3)$-torus link \\
    vortex type: hydrodynamic \\
    commutativity: non-commutative
\end{minipage}
\begin{minipage}[t]{0.6\linewidth}
    vortex charges:
    {\color{red} $%
\end{minipage}

\subsection{}
\begin{minipage}[t]{0.4\linewidth}
    file name: 14\_C\_2\_1\_1.mp4 \\
    phase: cyclic \\
    dynamics: relativistic \\
    knot type: simple ring
\end{minipage}
\begin{minipage}[t]{0.6\linewidth}
    vortex charges:
    $%
\end{minipage}

\subsection{}
\begin{minipage}[t]{0.4\linewidth}
    file name: 15\_C\_2\_1\_4.mp4 \\
    phase: cyclic \\
    dynamics: relativistic \\
    knot type: simple ring
\end{minipage}
\begin{minipage}[t]{0.6\linewidth}
    vortex charges:
    $%
\end{minipage}

\subsection{}
\begin{minipage}[t]{0.4\linewidth}
    file name: 16\_C\_2\_23\_45.mp4 \\
    phase: cyclic \\
    dynamics: relativistic \\
    knot type: trefoil \\
    commutativity: mutually non-commutative
\end{minipage}
\begin{minipage}[t]{0.6\linewidth}
    vortex charges:
    {\color{red} $%
\end{minipage}

\subsection{}
\begin{minipage}[t]{0.4\linewidth}
    file name: 17\_C\_2\_24\_12.mp4 \\
    phase: cyclic \\
    dynamics: relativistic \\
    knot type: $(2,4)$-torus link \\
    commutativity: mutually non-identical
\end{minipage}
\begin{minipage}[t]{0.6\linewidth}
    vortex charges:
    {\color{red} $%
\end{minipage}

\subsection{}
\begin{minipage}[t]{0.4\linewidth}
    file name: 18\_D4\_1\_24\_45.mp4 \\
    phase: $D_4$-biaxial nematic \\
    dynamics: non-relativistic \\
    knot type: $(2,4)$-torus link \\
    vortex type: hydrodynamic \\
    commutativity: MNIOPB
\end{minipage}
\begin{minipage}[t]{0.6\linewidth}
    vortex charges:
    {\color{red}    $%
\end{minipage}

\subsection{}
\begin{minipage}[t]{0.4\linewidth}
    file name: 19\_D4\_1\_24\_46.mp4 \\
    phase: $D_4$-biaxial nematic \\
    dynamics: non-relativistic \\
    knot type: $(2,4)$-torus link \\
    vortex type: hydrodynamic \\
    commutativity: MNIOPB
\end{minipage}
\begin{minipage}[t]{0.6\linewidth}
    vortex charges:
    {\color{red}    $%
\end{minipage}

\end{document}